\documentclass[a4paper]{article}
\usepackage{amsmath}
\usepackage[english]{babel}
\usepackage{latexsym}
\usepackage{amssymb}
\usepackage{amscd}
\usepackage{amsgen,amstext,amsbsy,amsopn}
\usepackage{math rsfs}
\usepackage{bm,bbm}
\usepackage{amsthm,epsfig,graphicx,graphics}
\usepackage[latin1]{inputenc}
\usepackage{xspace}
\usepackage{amsxtra}
\usepackage{color}
\usepackage{dsfont}

\usepackage[pdftex]{hyperref}

\usepackage{geometry}
\numberwithin{equation}{section}
\newcommand{\bdm}{\begin{displaymath}}
\newcommand{\edm}{\end{displaymath}}
\newcommand{\bdn}{\begin{eqnarray}}
\newcommand{\edn}{\end{eqnarray}}
\newcommand{\bay}{\begin{array}{c}}
\newcommand{\eay}{\end{array}}
\newcommand{\ben}{\begin{enumerate}}
\newcommand{\een}{\end{enumerate}}
\newcommand{\beq}{\begin{equation}}
\newcommand{\eeq}{\end{equation}}
\newcommand{\beqn}{\begin{eqnarray}}
\newcommand{\eeqn}{\end{eqnarray}}
\newcommand{\bml}[1]{\begin{multline} #1 \end{multline}}
\newcommand{\bmln}[1]{\begin{multline*} #1 \end{multline*}}

\newcommand{\lf}{\left}
\newcommand{\ri}{\right}
\newcommand{\bra}[1]{\lf\langle #1\ri|}
\newcommand{\ket}[1]{\lf|#1 \ri\rangle}

\newcommand{\rv}{\mathbf{r}}

\newcommand{\aae}{a_{\eps}}

\newcommand{\ate}{\tilde{a}_{\eps}}

\newcommand{\beps}{b_{k}}
\newcommand{\deps}{\delta_{\eps}}

\newcommand{\de}{d_{\eps}}
\newcommand{\diff}{\mathrm{d}}
\newcommand{\eps}{\varepsilon}

\newcommand{\dist}{\mathrm{dist}}

\newcommand{\ie}{I_{\eps}}

\newcommand{\bti}{\bar{t}_{\eps}}
\newcommand{\aldown}{\underline{\alpha}}
\newcommand{\alup}{\overline{\alpha}}
\newcommand{\kt}{\tilde{K}}
\newcommand{\ieps}{I_{\eps}}
\newcommand{\game}{\gamma_{\eps}}
\newcommand{\sigme}{\sigma_{\eps}}
\newcommand{\varre}{\varrho_{\eps}}

\newcommand{\gav}{\bm{\gamma}}
\newcommand{\nuv}{\bm{\nu}}

\newcommand{\ba}{\mathcal{B}}

\newcommand{\teps}{t_{\eps}}

\newcommand{\gle}{E^{\mathrm{GL}}}
\newcommand{\glm}{\Psi^{\mathrm{GL}}}

\newcommand{\gldom}{\mathscr{D}^{\mathrm{GL}}}
\newcommand{\aav}{\mathbf{A}}
\newcommand{\aavm}{\mathbf{A}^{\mathrm{GL}}}

\newcommand{\hex}{b}
\newcommand{\theo}{\Theta_0}

\newcommand{\glfk}{\mathcal{G}_{\kappa,\sigma}^{\mathrm{GL}}}

\newcommand{\glfe}{\mathcal{G}_{\eps}^{\mathrm{GL}}}
\newcommand{\glee}{E_{\eps}^{\mathrm{GL}}}

\newcommand{\annf}{\mathcal{G}_{\ann}}
\newcommand{\annft}{\tilde{\mathcal{G}}_{\ann}}
\newcommand{\curv}{k(s)}

\newcommand{\fal}{\fkal}

\newcommand{\hkal}{H_{k,\al}}
\newcommand{\phikal}{\phi_{k,\alpha}}

\newcommand{\pot}{V_{\eps,\alpha}}

\newcommand{\mue}{\mu_{\eps}(k,\alpha)}

\newcommand{\muosc}{\mu^{\mathrm{osc}}(\alpha)}

\newcommand{\hosc}{H^{\mathrm{osc}}}

\newcommand{\curl}{\mbox{curl}}

\newcommand{\ann}{\mathcal{A}_{\eps}}
\newcommand{\annt}{\tilde{\mathcal{A}}_{\eps}}
\newcommand{\annd}{\mathcal{A}_{\rm bl}}
\newcommand{\anntd}{\tilde{\mathcal{A}}_{\rm bl}}

\newcommand{\trial}{\Psi_{\mathrm{trial}}}
\newcommand{\atrial}{\mathbf{A}_{\mathrm{trial}}}

\newcommand{\half}{\mbox{$\frac{1}{2}$}}

\newcommand{\disp}{\displaystyle}
\newcommand{\tx}{\textstyle}

\newcommand{\Z}{\mathbb{Z}}
\newcommand{\R}{\mathbb{R}}

\newcommand{\C}{\mathbb{C}}

\newcommand{\E}{\mathcal{E}}

\newcommand{\OO}{\mathcal{O}}

\newcommand{\al}{\alpha}

\newcommand{\ep}{\varepsilon}

\newcommand{\Om}{\Omega}

\newcommand{\dd}{\partial}

\newcommand{\supp}{\mathrm{supp}}

\newcommand{\Hc}{H_{\mathrm{c}1}}
\newcommand{\Hcc}{H_{\mathrm{c}2}}
\newcommand{\Hccc}{H_{\mathrm{c}3}}

\newtheorem{teo}{Theorem}[section]
\newtheorem{lem}{Lemma}[section]
\newtheorem{pro}{Proposition}[section]
\newtheorem{cor}{Corollary}[section]

\newcounter{remark}[section]
\newenvironment{rem}{\refstepcounter{remark} \vspace{0,1cm} \noindent \textit{Remark \thesection.\theremark}\,}{\hfill \qed \vspace{0,2cm}}

\newcommand{\annbb}{\bar{\mathcal{A}}_{k,\eps}}
\newcommand{\annbk}{\bar{I}_{k,\eps}}
\newcommand{\keps}{\beta_{\eps}}
\newcommand{\btik}{\bar{t}_{k,\eps}}

\newcommand{\annpre}{\hat{\mathcal{A}}_\eps}
\newcommand{\annfpre}{\hat{\mathcal{G}}_\eps}

\newcommand{\fkal}{f_{k,\alpha}}
\newcommand{\fk}{f_{k}}
\newcommand{\fOal}{f_{0,\alpha}}
\newcommand{\fO}{f_0}

\newcommand{\fone}{\E^{\mathrm{1D}}}
\newcommand{\fonekal}{\E ^{\rm 1D}_{k,\alpha}}
\newcommand{\foneOal}{\E ^{\rm 1D}_{0,\alpha}}
\newcommand{\fonek}{\E ^{\rm 1D}_{k}}

\newcommand{\eone}{E^{\mathrm{1D}}}
\newcommand{\eonekal}{E ^{\rm 1D}_{k,\alpha}}
\newcommand{\eoneoal}{E ^{\rm 1D}_{0,\alpha}}
\newcommand{\eonek}{E ^{\rm 1D}_{k}}
\newcommand{\eoneo}{E ^{\rm 1D}_{0}}

\newcommand{\alk}{\alpha_k}
\newcommand{\alkt}{\tilde{\alpha}_k}
\newcommand{\alO}{\alpha_0}
\newcommand{\Fk}{F_k}
\newcommand{\Kk}{K_k}

\newcommand{\FO}{F_0}
\newcommand{\KO}{K_0}

\newcommand{\potkal}{V_{k,\alpha}}
\newcommand{\potOal}{V_{0,\alpha}}
\newcommand{\potk}{V_{k}}
\newcommand{\potO}{V_{0}}

\newcommand{\annfO}{\mathcal{F}_{\ann}}
\newcommand{\aaO}{a_{0}}
\newcommand{\es}{\mathbf{e}_s}

\newcommand{\jv}{\textbf{j} (v)}
\newcommand{\muv}{\mu (v)}

\newcommand{\Ehp}{\mathcal{E}_{\rm hp}}
\newcommand{\Ehpe}{E_{\rm hp}}

\begin{document}

\markboth{\scriptsize{\textsc{Correggi, Rougerie} -- Surface Superconductivity}}{\scriptsize{\textsc{Correggi, Rougerie} -- Surface Superconductivity}}

\title{On the Ginzburg-Landau Functional in the Surface Superconductivity Regime}

\author{M. Correggi${}^{a}$, N. Rougerie${}^{b}$
	\\
	\normalsize\it ${}^{a}$ Dipartimento di Matematica, ``Sapienza'' Universit\`{a} di Roma,	\\
	\normalsize\it Piazzale Aldo Moro 5, 00185, Rome, Italy.	\\
	\normalsize\it ${}^{b}$ Universit\'e de Grenoble 1 \& CNRS, LPMMC \\ 
	\normalsize\it Maison des Magist\`{e}res CNRS, BP166, 38042 Grenoble Cedex, France.}
	
\date{November 19, 2014}

\maketitle

\begin{abstract} 
We present new estimates on the two-dimensional Ginzburg-Landau energy of a type-II superconductor in an applied magnetic field varying between the second and third critical fields. In this regime, superconductivity is restricted to a thin layer along the boundary of the sample. We provide new energy lower bounds, proving that the Ginzburg-Landau energy is determined to leading order by the minimization of a simplified 1D functional in the direction perpendicular to the boundary. Estimates relating the density of the Ginzburg-Landau order parameter to that of the 1D problem follow. In the particular case of a disc sample, a refinement of our method leads to a pointwise estimate on the Ginzburg-Landau order parameter, thereby proving a strong form of uniformity of the surface superconductivity layer, related to a conjecture by Xing-Bin Pan.    
\end{abstract}

\tableofcontents

\section{Introduction}

The Ginzburg-Landau (GL) theory of superconductivity was first introduced in the '50s \cite{GL} as a phenomenological, macroscopic model. It was later justified by Gorkov \cite{Gor} as emerging from the microscopic Bardeen-Cooper-Schrieffer (BCS) theory \cite{BCS} (a rigorous mathematical derivation has been achieved only very recently \cite{FHSS}). It has been widely used in the physics literature (see the monographs \cite{Le,Ti}) and  has proved very successful, e.g., in predicting the response of superconducting materials to an external magnetic field. We recall for instance the  prediction by Abrikosov\footnote{Noteworthily Ginzburg, Landau and Abrikosov were all awarded a Nobel prize in physics for their findings about superconductivity.} of vortex lattices \cite{Ab} and the first discussion by Saint-James and de Gennes \cite{sJdG} of the surface superconductivity phenomenon that shall concern us here.

The phenomenological quantities associated with the superconductor are an order parameter $ \Psi $, such that $ |\Psi|^2 $ measures the relative density of superconducting Cooper pairs, and an induced magnetic field $ h  $, which must be distinguished from the external or applied magnetic field. For a superconductor confined to an infinite cylinder of smooth cross section $\Om \subset \R ^2$, the GL free energy is given in appropriate units by the functional (here we follow the convention of \cite{FH}, other choices are possible, see, e.g., \cite{SS})
\beq\label{eq:gl func}
	\glfk[\Psi,\aav] = \int_{\Om} \diff \rv \: \lf\{ \lf| \lf( \nabla + i \kappa\sigma  \aav \ri) \Psi \ri|^2 - \kappa^2 |\Psi|^2 + \half \kappa^2 |\Psi|^4 + \lf(\kappa \sigma \ri)^2 \lf| \curl \aav - 1 \ri|^2 \ri\},
\eeq
with $\Psi:\R ^2 \to \C$ the order parameter and $ \kappa \sigma \aav:\R ^2 \to \R ^2 $ the induced magnetic vector potential. The induced magnetic field is then $ h = \kappa \sigma  \: \curl \aav $, with $ \kappa >0  $ a physical parameter (penetration depth) characteristic of the material, and $\kappa \sigma  $ measures the intensity of the external magnetic field, that we assume to be constant throughout the sample. We shall be concerned with type-II superconductors, characterized by $\kappa > 1/\sqrt{2}$, and more precisely with the limit $\kappa \to \infty$ (extreme type-II).  


The state of the superconductor is obtained by minimizing the GL free energy with respect to the pair $ (\Psi, \aav) $. A crucial property to be taken into account in the minimization is the \emph{gauge invariance} of \eqref{eq:gl func}, namely the fact that the energy does not change when the replacements 
\begin{equation}\label{eq:gauge inv}
\Psi \to \Psi e^{-i \kappa \sigma \varphi}, \qquad \aav \to \aav + \nabla \varphi 
\end{equation}
are simultaneously performed, for some (say smooth) function $\varphi:\R ^2 \to \R$. One important consequence is that the only physically meaningful quantities are the gauge invariant ones, such as the induced field $h= \kappa \sigma \: \curl \aav $ or the density of Cooper pairs $|\Psi| ^2$.


The modulus of the order parameter $ |\Psi| $ varies between $ 0 $ and $ 1 $: the vanishing of $ \Psi $ in a certain region or point implies a loss of the superconductivity there due to the absence of Cooper pairs, whereas if $ |\Psi| = 1$ somewhere all the electrons are arranged in Cooper pairs and thus superconducting. The cases $ |\Psi| \equiv 1 $ and $ |\Psi| \equiv 0 $ everywhere in $ \Omega $ correspond to the so-called {\it perfectly superconducting} and {\it normal} states, known to be preferred for small and large applied field respectively. When $|\Psi|$ is not identically $ 0 $ nor $ 1 $, for intermediate values of the applied field, one says that the system is in a {\it mixed} state.

Several remarkable phenomena occur in between the two extreme regimes where the superconducting or normal state dominate. We describe them in order of increased applied field, that is, in the units of \eqref{eq:gl func}, in order of increasing $\kappa \sigma$  (recall that this description is valid in the regime~$\kappa\gg 1$): 
\begin{itemize}
\item The sample stays in the superconducting state until the \textit{first critical field} 
\[
\Hc = C_{\Omega} \log \kappa 
\]
is reached, where $C_{\Omega}$ only depends on the domain $\Om$. Then vortices, i.e., isolated zeros of $\Psi$ with non-trivial winding number (phase circulation), start to appear in the GL minimizer and their number increases with the increase of the external field $\kappa \sigma $. They arrange themselves on a triangular lattice, the famous \textit{Abrikosov lattice}, that survives until a second critical value of the field is reached. 
\item When the \textit{second critical field}
\beq
	\label{eq:critical field 2}
	\Hcc = \kappa ^2,	
\eeq
is crossed, superconductivity is lost in the bulk of the sample and survives only in a thin layer at the boundary $ \partial \Omega $. More precisely the GL order parameter is exponentially decaying far from the boundary and well separated from $0 $ only up to distances of order $(\kappa \sigma) ^{-1/2}$ from $ \partial \Omega $. This is the {\it surface superconductivity} regime: see \cite{SPSKC} for an early observation of this phenomenon and \cite{NSG} for more recent experimental data.
\item Surface superconductivity survives until a \textit{third critical field} 
\beq
	\label{eq:critical field 3}
	\Hccc  = \theo^{-1} \kappa ^2 + \OO(1),
\eeq
is reached, where $ 1/2 < \theo < 1 $ (more precisely $ \theo^{-1} \simeq 1.6946 $) is a sample-independent number (see Section \ref{1d So: sec}). Above $ \Hccc $ there is a total loss of superconductivity everywhere and the normal state becomes the global minimizer of the GL energy. Sample-dependent corrections to the estimate \eqref{eq:critical field 3} are also known.
\end{itemize}
The above is of course a vague description and a large mathematical literature, including the present contribution, has been devoted to backing the above heuristics with rigorous results (see \cite{BBH2,FH,SS,Sig} for reviews and references). On the experimental side, we mention the direct imaging of the Abrikosov lattices (see, e.g., \cite{Hetal}) and more recently of the surface superconductivity state \cite{NSG}.

\medskip

In this paper we investigate the behavior of the superconducting sample between the second and third critical fields, which in our units translates into the assumption 
\beq
	\label{eq:external field}
	\sigma = b \kappa
\eeq
for some fixed parameter $b$ satisfying the conditions
\beq
	\label{eq:b condition}
	1 < b < \theo^{-1}.
\eeq
From now on we introduce more convenient units to deal with the surface superconductivity phenomenon: we define the small parameter 
\beq
	\label{eq:eps}
	\eps = \frac{1}{\sqrt{\sigma \kappa}} = \frac{1}{b^{1/2}\kappa} \ll 1 
\eeq
and study the asymptotics $ \eps \to 0 $ of the minimization of the functional \eqref{eq:gl func}, which in the new units reads
\beq\label{eq:GL func eps}
	\glfe[\Psi,\aav] = \int_{\Om} \diff \rv \: \bigg\{ \bigg| \bigg( \nabla + i \frac{\aav}{\eps^2} \bigg) \Psi \bigg|^2 - \frac{1}{2 \hex \eps^2} \lf( 2|\Psi|^2 - |\Psi|^4 \ri) + \frac{\hex}{\eps^4} \lf| \curl \aav - 1 \ri|^2 \bigg\}.
\eeq
We shall denote
\beq
	\glee : = \min_{(\Psi, \aav) \in \gldom} \glfe[\Psi,\aav],
\eeq
with\footnote{Sometimes the minimization domain is defined in a slightly different way by picking divergence-free magnetic fields $ \aav $, i.e., such that $ \nabla \cdot \aav = 0 $ in $ \Omega $ and $ \nuv \cdot \aav = 0 $ on $ \partial \Omega $. The gauge invariance of the functional immediately implies that the two minimization domains are in fact equivalent.}
\beq
	\gldom : = \lf\{ (\Psi,\aav) \in H^1(\Om;\C) \times H^1(\Om;\R^2) \ri\},
\eeq
and denote by $ (\glm,\aavm) $ a minimizing pair (known to exist by standard methods \cite{FH,SS}). Minimizers of the GL functional solve the GL variational equations\footnote{$\Im [ \: \cdot \:]$ stands for the imaginary part.}
\beq
	\label{eq:GL eq}
	\begin{cases}
		- \lf( \nabla + i \frac{\aavm}{\eps^2} \ri)^2 \glm = \frac{1}{\hex \eps^2} \lf( \lf| \glm \ri|^2 - 1 \ri) \glm,	\\
		\curl \lf( \curl \aavm \ri) = \eps^2 \Im \lf[ {\glm}^* \lf( \nabla + i \frac{\aavm}{\eps^2} \ri) \glm \ri],	
	\end{cases}
	\qquad \mbox{in } \Om,	
\eeq
with the boundary conditions
\beq
	\label{eq:GL bc}
	\begin{cases}
		\nuv \cdot \lf( \nabla + i \frac{\aavm}{\eps^2} \ri)  \glm = 0,	\\
		\curl \aavm = 1,	
	\end{cases}
	\qquad \mbox{on } \partial \Om,
\eeq
where $\nuv$ is the outward unit normal to the boundary.

Note that in our new choice of units we have replaced the two parameters $ \kappa, \sigma $ with the pair $ \eps, b $, where only the first one is infinitesimal throughout the region between $ \Hcc $ and $ \Hccc $. The strict inequalities we impose on the fixed parameter $b$ in \eqref{eq:b condition} are in fact crucial in our analysis: we shall not deal with the limiting cases where $b\to \theo^{-1}$ or $b\to 1$ as a function of $\eps$. The former corresponds to applied fields extremely close to $\Hccc$, a regime already thoroughly covered in the literature (see \cite[Chapters 13 and 14]{FH}  and references therein). The latter corresponds to the transition from boundary to bulk behavior where our methods do not apply (see \cite{FK,Kac} for recent results).

\medskip

Several important facts about type-II superconductors in the surface superconductivity regime $\eps \to 0$, $1<b<\theo ^{-1}$, have been proved rigorously in recent years. We start by mentioning the energy asymptotics of \cite{Pan}: 
\beq
	\label{eq:FH energy asympt}
	\glee = \frac{|\dd \Om| E_{\hex}}{\eps} + o(\eps^{-1}),
\eeq
where $ E_{\hex} < 0 $ is some constant independent of $ \eps $, $|\dd \Om | $ the length of the boundary of $\Om$ and $ b $ satisfies the conditions \eqref{eq:b condition}. The definition of $E_{\hex}$ given in \cite{Pan} is somewhat complicated (see Section \ref{sec:sketch}) and further research has been devoted to obtaining a simplified expression. First, \cite{FH1} considered the case where $b\to \theo ^{-1}$ as a function of $\eps$ at a suitable rate and identified the constant $E_{\hex}$ in this case (see also \cite{LP}). In \cite{AH}, on the other hand, it is proved that if $ b $ is sufficiently close to $ \theo^{-1} $ (independently of $\eps$), the following holds:
\beq
	\label{eq:FH energy asympt 2}
	\glee = \frac{|\dd \Om| \eoneo}{\eps} + \OO(1),
\eeq
where $\eoneo$ is obtained by minimizing the functional
\begin{equation}\label{eq:intro 1D func}
\fone_{0,\alpha}[f] : = \int_0^{+\infty} \diff t \lf\{ \lf| \partial_t f \ri|^2 + (t + \alpha )^2 f^2 - \frac{1}{2b} \lf(2 f^2 - f^4 \ri) \ri\} 
\end{equation}
both with respect to the function $f$ and the real number $\alpha$. All the previous results are reproduced in \cite[Theorem 14.1.1]{FH}. It was later proved in \cite{FHP} that \eqref{eq:FH energy asympt 2} holds at least for $1.25 \leq b \leq  \theo ^{-1} $, which is an improvement over \cite{AH}, but does not cover the full surface superconductivity regime~\eqref{eq:b condition}.

The idea behind \eqref{eq:FH energy asympt 2} is that, up to a suitable choice of gauge, any minimizing order parameter $\glm$ for \eqref{eq:gl func} has the structure 
\begin{equation}\label{eq:GLm structure formal}
\glm(\rv) \approx \fO \left(\tx \frac{\eta}{\eps} \right)  \exp \left( - i \alO \tx \frac{\xi}{\eps}\right) \exp \left(  i \phi_\eps \left(\tx\frac{\xi}{\eps},\tx\frac{\eta}{\eps}\right) \right)
\end{equation}
where $(\fO,\alO)$ is a minimizing pair for \eqref{eq:intro 1D func} and $(\xi,\eta)=$(tangent coordinate, normal coordinate) are boundary coordinates defined in a tubular neighborhood of $\dd \Om$ (see Section \ref{sec:restriction annulus}). The additional phase factor $\exp \left(  i \phi_\eps (\xi,\eta) \right)$ is less relevant in that a suitable change of gauge will make it disappear to obtain an effective problem in terms of the remaining part of the wave function. See Remark~2.\ref{rem:winding} below for its precise definition. Of course the only physically relevant quantities are gauge invariant and thus the actual result, proven in \cite{AH}, can be stated as   
\begin{equation}\label{eq:GLm structure}
\left\Vert |\glm| ^2 - \left| \fO \left( \tx\frac{\eta}{\eps} \right) \right| ^2 \right\Vert_{L ^2 (\Om)} \ll \left\Vert \fO ^2 \left( \tx\frac{\eta}{\eps} \right)  \right\Vert_{L ^2 (\Om)},
\end{equation}
which says that the density of superconducting electrons in the sample is essentially confined to the boundary on a length scale $  \eps = 1/\sqrt{\kappa \sigma}$. Moreover it can be approximated by the function $\left| \fO \left(\tx\frac{\eta}{\eps}\right) \right| ^2 $ which only depends on the normal coordinate. The result of \cite{AH} is not stated exactly as above, but following their methods and those of \cite{FHP}, it is possible to see that \eqref{eq:GLm structure} holds for $1.25 \leq b  < \theo^{-1} $.

\medskip

Two main questions are left open in the aforementioned contributions:
\begin{enumerate}
\item Does \eqref{eq:FH energy asympt 2} hold in the full surface superconductivity regime \eqref{eq:b condition}?
\item Can \eqref{eq:GLm structure} be strengthened to a better norm, e.g., is the GL matter density $|\glm| ^2$ close to the simplified 1D density $\left| \fO \left(\tx\frac{\eta}{\eps}\right) \right| ^2$ in $L ^{\infty}$ norm, at least close to the boundary of the sample?
\end{enumerate}
These are respectively advertised as Open Problems number 2 and 4 in the list of \cite[Page 267]{FH}. Both questions are important from a physical point of view. An affirmative answer to Question~1 would rigorously confirm that surface superconductivity is essentially a 1D phenomenon in the direction normal to the boundary. Question 2 is a strengthened version of a conjecture due to X.B. Pan \cite[Conjecture 1]{Pan}: it asks whether the surface superconducting layer is uniform in some sense, in particular by ruling out normal inclusions (vanishing of the order parameter) like isolated vortices close to the boundary of the sample. 

\medskip

The goal of the present paper is to provide a method for bounding below the GL energy that allows us to answer Question 1 in the affirmative. In the particular case where the domain $\Om$ is a disc, a refined version of the method and some specific technical efforts also give a positive answer to Question 2. 

Our new estimates follow from a quite different approach than that in \cite{AH,FHP}. They are inspired by our earlier works on the related Gross-Pitaevskii (GP) theory of trapped rotating superfluids \cite{CPRY2,CPRY3} (see \cite{CPRY4,CPRY5} for short presentations and \cite{R1} for an earlier approach). These were concerned with the occurrence or absence of vortices in the bulk of rotating Bose-Einstein condensates, in particular in the giant vortex regime where the wave-function of the condensate is concentrated in a thin annulus. 

The main physical insight behind our new results on the GL functional is roughly speaking that the only possible way for \eqref{eq:FH energy asympt 2} to fail would be that vortices are nucleated inside the surface superconductivity layer. Although this is a physically quite unreasonable possibility in view of the accumulated knowledge on type-II superconductors, it is not ruled out by any previous mathematical result. The method we adapt from our earlier works confirms that nucleating vortices is not favorable, which allows us to prove \eqref{eq:FH energy asympt 2} in the whole surface superconductivity regime. 

A sketch of the method will be given in Subsection \ref{sec:sketch}, after we state our main results rigorously: first, our affirmative answer to Question 1 for general domains in Subsection \ref{sec:main generic}, then our investigation of Question 2 for disc samples in Subsection \ref{sec:main disc}. The rest of the paper is devoted to the proofs of our main results.

\medskip

\noindent\textbf{Notation:} In the whole paper, $C$ will denote a generic positive constant whose value may change from line to line. The notation $\OO(\delta)$ and $o(\delta)$ will as usual denote quantities whose absolute value is respectively bounded by $C\delta$ or a function $f(\delta)\to 0$ in the relevant limit at hand (most often, $\ep \to 0$). We will write $\OO(\delta ^{\infty})$ for a quantity which is a $\OO(\delta ^k)$ for any $k\in \R^+$ when $\delta \to 0$, e.g., an exponentially small quantity.  We will use $a\sim b$, $a\ll b$ and $a\propto b$ when $a/b\to 1$, $a/b \to 0 $ and $a / b \to C \neq 0,1$ respectively. We sometimes use $f\approx g$ for two functions $f$ and $g$ in heuristic statements, when we do not wish to be precise about the norm in which $f$ and $g$ are close  nor about the errors involved.

\medskip

\noindent\textbf{Acknowledgments:} Part of this research has been carried out at the \textit{Erwin Schr\"odinger Institute} in Vienna and at the \textit{Institut Henri Poincar\'e} in Paris. We acknowledge stimulating discussions with S\o{}ren Fournais. M.C. was supported by the European Research Council under the European Community Seventh Framework Program (FP7/2007-2013 Grant Agreement CoMBos No. 239694). This research also received support from the ANR (Project Mathostaq ANR-13-JS01-0005-01).

\section{Main Results}\label{sec:main results}

\subsection{General domains: leading order of the energy}\label{sec:main generic}

Here we present the results we may prove for any smooth domain $\Om$. Minimizing the 1D functional \eqref{eq:intro 1D func} with respect to $f$ we obtain an energy $\eoneoal$ and a minimizer $\fOal$. Then, minimizing $\eoneoal$ with respect to $\alpha$ gives a minimal energy $\eoneo$  and a minimizer $\alO$. We denote $\fO : = f_{0,\alO}$ for short. The intuition behind \eqref{eq:FH energy asympt 2} is that $\glm$ behaves to leading order as \eqref{eq:GLm structure formal} in a suitable gauge. We use again the notation $(\xi,\eta)$ for the tangential and normal components of boundary coordinates  (see \cite[Appendix F]{FH}). Our main result for general domains is the following answer to Question 1 and extension of \eqref{eq:GLm structure}:

\begin{teo}[\textbf{Leading order of the energy and density for general domains}]\label{theo:generic}\mbox{}\\
Let $\Om\subset \R ^2$ be any smooth simply connected domain. For any fixed $1<b<\theo ^{-1}$, in the limit $ \eps \to 0$, it holds
\begin{equation}\label{eq:main energy generic}
\glee = \frac{|\dd \Om| \eoneo}{\eps} + \OO(1),
\end{equation} 
\begin{equation}\label{eq:main density generic}
\left\Vert |\glm| ^2 -  \fO ^2 \left(\tx\frac{\eta}{\eps} \right)  \right\Vert_{L ^2 (\Om)} \leq C \eps \ll \left\Vert \fO ^2 \left(\tx\frac{\eta}{\eps}\right)  \right\Vert_{L ^2 (\Om)}. 
\end{equation}
\end{teo}

The estimate \eqref{eq:main density generic} confirms that the density of Cooper pairs is roughly constant in the direction parallel to the boundary of the sample, while the profile in the perpendicular direction is to leading order independent of the sample $\Om$. Note that we are making a slight abuse of notation in this estimate since the boundary coordinates are not defined in the whole sample. This is harmless since both $|\glm| ^2$ and $\fO ^2 \left(\tx\frac{\eta}{\eps} \right)$ are exponentially decaying at distances larger than $C \eps$ from the boundary. This remark also indicates that 
$$\left\Vert \fO ^2 \left(\tx\frac{\eta}{\eps}\right)  \right\Vert_{L ^2 (\Om)} \propto \eps ^{1/2}, $$
because the area of the boundary layer is roughly proportional to $\eps$. This vindicates the second inequality in \eqref{eq:main density generic}.

\begin{rem}(Limiting cases)\label{rem:lim case}\mbox{}\\
As already mentioned we choose not to address the limiting cases $b\to 1$ or $b\to \theo^{-1}$ for simplicity. In the former case, a bulk term should appear in addition to the boundary term \eqref{eq:main energy generic}, as demonstrated in \cite{FK}. Probably our method may give a simplified expression of the boundary term obtained in this reference (see Subsection \ref{sec:sketch}). When $b\to \theo^{-1}$, the analysis is complicated by the fact that $\fO \to 0$, so that the leading order of $\glee$ is no longer of order $\eps ^{-1}$ \cite[Theorem 14.1.1]{FH}.   
\end{rem}

\subsection{The case of the disc: refined energy estimates and Pan's conjecture}\label{sec:main disc}

Although the proof of Theorem \ref{theo:generic} suggests that normal inclusions, i.e., regions where the superconductivity is lost (isolated vortices for example), are not favorable in the surface superconductivity layer, the precision of the energy estimate is not sufficient to rigorously conclude that none occur. In fact the error term in \eqref{eq:main energy generic} is still comparable with the energetic cost for $|\glm|$ to vanish close to the boundary in a ball of radius $\eps$, as follows from the optimal estimate $|\nabla |\glm|| \propto \eps ^{-1}$. In order to obtain an equivalent of \eqref{eq:GLm structure} in $L ^{\infty}$ norm, it is thus necessary to expand the energy further.

An important difficulty is then that one of the terms entering the $\OO(1)$ remainder in \eqref{eq:main energy generic} involves the curvature of the domain $\Om$. It is clearly (see below) not smaller than a constant, and must thus be evaluated to go beyond \eqref{eq:main energy generic}. One case where we are able to do this is when the curvature is a constant $k$, i.e., the domain is a disc. In this case one may evaluate the $\OO(1)$ remainder in \eqref{eq:main energy generic} by using the following refined functional: 
\begin{equation}
\label{eq:intro 1D func disc}
\fone_{k,\alpha}[f] : = \int_0^{c_0|\log\eps|} \diff t \: (1-\eps k t )\lf\{ \lf| \partial_t f \ri|^2 + \frac{(t + \alpha - \frac12 \eps k t ^2 )^2}{(1-\eps k t ) ^2} f^2 - \frac{1}{2b} \lf(2 f^2 - f^4 \ri) \ri\},
\end{equation}
where $ c_0 $ is a constant that has to be chosen large enough (see Subsection \ref{sec:restriction annulus} for its role in the proof). Note that \eqref{eq:intro 1D func} is simply the above functional with $k=0$, $\eps=0$, that is with curvature terms neglected.
As before, minimizing with respect to both $f$ and $\alpha$, we obtain an energy $\eonek$, a function $\fk$ and an optimal $\alk\in \R$. Contrarily to the corresponding quantities in the $k=0$ case, these do depend on~$\eps$. The following theorem contains the refined estimates in which we use \eqref{eq:intro 1D func disc}. We denote by $(r,\vartheta)$ the polar coordinates. 

\begin{teo}[\textbf{Refined energy and density estimates in the disc case}]
\label{theo:disc energy}
\mbox{}	\\
Let $\Om$ be a disc of radius $R = k ^{-1}$. For any $ 1 < \hex < \theo^{-1} $ there exists a $c_0 >0$ such that, as $ \eps \to 0$, it holds
\beq
\label{eq:main energy disc}
\glee = \frac{2\pi \eonek}{\eps} + \OO(\eps |\log\eps|),
\eeq
\begin{equation}\label{eq:main density disc}
\left\Vert |\glm| ^2 -  \fk ^2 \left(\tx\frac{R-r}{\eps}\right)  \right\Vert_{L ^2 (\Om)} = \OO(\eps^{3/2} |\log \eps| ^{1/2}). 
\end{equation}
\end{teo}

\begin{rem}(Limits and orders of magnitude)\label{rem:orders}\mbox{}\\
It is clear from \eqref{eq:intro 1D func disc} that the $k$-dependent terms contribute to $\eonek$ at order $\eps$, so that the error term in \eqref{eq:main energy disc} is much smaller than these curvature dependent effects. Note also that the error we make in the density asymptotics \eqref{eq:main density disc} is much smaller than the difference between $\fO$ and $\fk$ (of order $\eps$ as a simple estimate reveals). In particular \eqref{eq:main density disc} does \emph{not} hold if $\fk$ is replaced by $\fO$.  
\end{rem}

A density estimate so strong as \eqref{eq:main density disc} turns out to be incompatible with any significant local discrepancies between $|\glm| ^2$ and  $\fk ^2 \left(\tx\frac{R-r}{\eps} \right)$. We thus rule out normal inclusions in some boundary layer that we now define: 
\beq
\label{annd}
\annd : = \lf\{ \rv \in \Om \: : \: \fk \lf(\tx\frac{R-r}{\eps} \ri) \geq \game \ri\} \subset \lf\{ r \geq R - \half \eps \sqrt{|\log\game|} \ri\},
\eeq
where $\rm{bl}$ stands for ``boundary layer'' and $ 0 < \game \ll 1 $ is any quantity such that
\beq
\label{game}
\game \gg \eps^{1/6}|\log \eps| ^{4/3}.
\eeq
The inclusion in \eqref{annd} follows from \eqref{fal point l u b} below and ensures we are really considering a significant boundary layer: recall that the physically relevant region has a thickness roughly of order $\eps$.

	\begin{teo}[\textbf{Uniform density estimates in the disc case}]
		\label{theo:Pan}
		\mbox{}	\\
		Let $\Om$ be a disc of radius $R = k ^{-1}$. For any $ 1 < \hex < \theo^{-1} $ there exists a $c_0 >0$ such that, as $ \eps \to 0$, it holds
		\beq
			\label{eq:Pan plus}
			\lf\| \lf|\glm(\rv)\ri| - \fk\lf(\tx\frac{R-r}{\eps} \ri) \ri\|_{L^{\infty}(\annd)} = \OO(\game^{-3/2} \eps^{1/4} |\log\eps|^{2}) \ll 1.
		\eeq
		In particular at the boundary $ \partial \Omega $ we have
		\begin{equation}\label{eq:Pan}
			\lf| \lf| \glm(R,\vartheta) \ri| - \fk (0) \ri| =  \OO(\eps^{1/4} |\log\eps|^{2}),
		\end{equation}
		uniformly in $\vartheta \in [0,2\pi]$.
	\end{teo}

	\begin{rem}(Optimal density $ \fk $ and Pan's conjecture)
		\mbox{}	\\
		For example, if we choose $ \game = |\log\eps|^{-2} $, we obtain
		\bdm
			\lf|  \lf|\glm(\rv)\ri| - \fk\lf(\tx\frac{R-r}{\eps}\ri) \ri| \leq C \eps^{1/4}|\log\eps|^5,
		\edm
		for any $\rv$ such that $r\geq R - C\eps \log |\log \eps|$. One should note that the error term in \eqref{eq:Pan plus} is much larger than the difference between $\fk$ and $\fO$, so we could as well use $\fO$ in the statement of the theorem. However the use of the optimal density $ \fk $ remains crucial to reproduce the energy of the GL minimizer up to corrections $ o(1) $ (see Proposition \ref{splitting}). The particular case \eqref{eq:Pan} proves the original form of Pan's conjecture \cite[Conjecture 1]{Pan} in the case of disc samples. 
	\end{rem}

	The above result provides information about the behavior of the modulus of $ \glm $ on the boundary, and, thanks to the positivity of $ \fk $ (in particular at $ t = 0 $), the phase and thus the winding number (phase circulation) of $ \glm $ at the boundary are well defined. In the next theorem we give an estimate of $ \deg(\glm, \partial \Omega) $, where, if $ \ba_R $ is a ball of radius $ R $, we define
	\beq\label{eq:GL degree}
		2 \pi \deg\lf(\Psi, \partial \ba_R\ri) : = - i \int_{\partial \ba_R} \diff \xi \: \frac{|\Psi|}{\Psi} \partial_{\tau} \lf( \frac{\Psi}{|\Psi|} \ri),
	\eeq
	$ \partial_{\tau} $ standing for the tangential derivative along the boundary. We thus complete the analysis of the boundary behavior of $ \glm $ by discussing its phase.

	\begin{teo}[{\bf Winding number of $ \glm $ on the boundary of disc samples}]
		\label{teo:circulation}
		\mbox{}	\\
		Let $\Om$ be a disc of radius $R = k ^{-1}$. For any $ 1 < \hex < \theo^{-1} $ as $ \eps \to 0$, any GL minimizer $ \glm $ satisfies
		\beq\label{eq:GL degree result}
			\deg\lf(\glm, \partial \Omega\ri) = \frac{\pi R^2}{\eps^2} +  \frac{|\alk|}{\eps} + \OO(\eps^{-3/4}|\log\eps|^2).
		\eeq
	\end{teo}

	\begin{rem}(Boundary behavior of $ \glm $ and interpretation of the winding number)\label{rem:winding}
		\mbox{}	\\
		The combination of Theorem \ref{teo:circulation} with the modulus convergence stated in \eqref{eq:Pan} is compatible with a behavior of the type~\eqref{eq:GLm structure formal}, with $\phi_\eps$ being given by 
		$$
			\phi_{\eps}(s,t) : = - \frac{1}{\eps} \int_{0}^{t} \diff \eta \: \nuv(\eps s) \cdot \aavm(\rv(\eps s, \eps \eta)) + \frac{1}{\eps} \int_{0}^s \diff \xi \: \gav^{\prime}(\eps \xi) \cdot \aavm(\rv(\eps \xi,0)) - \eps \deps s.			$$	
		    with $\aavm$ the minimizing vector potential, $ \nuv $ and $ \gav $ unit vectors respectively normal and tangential to the boundary and $ \deps $ proportional to the circulation of $ \aavm $ on $ \partial \Omega $ (see Section~\ref{sec:restriction annulus} for a more precise discussion). On the boundary of the disc this boils down to
		\bdm
			\glm(R,\vartheta) \approx \fk(0) \exp\lf\{i \lf[ \tx\frac{\pi R^2}{\eps^2} + \frac{|\alk|}{\eps} \ri] \vartheta \ri\},
		\edm
		although the estimates are not precise enough to derive the exact shape of the phase factor. As is well-known, the winding number counts the number of phase singularities, i.e., vortices of $\glm$ inside $\Om$. The above result thus suggests that there are  $\frac{\pi R^2}{\eps^2} +  \frac{|\alk|}{\eps}$ vortices in the sample, although these are made undetectable in practice by the exponential decay of $\glm$ in this region.   
	\end{rem}

	\begin{rem}(Lack of continuity of $\glm$ as a function of $b,\eps$)\mbox{}\\
	 As noted in \cite{AH}, the continuity of $\glm$ as a function of the parameters of the functional does not seem compatible with the existence and quantization of the winding number \eqref{eq:GL degree} that follows from \eqref{eq:Pan}. Indeed, the phase circulation is a topological, discrete quantity that cannot vary continuously. While in \cite{AH} this argument was used to cast doubt on the possibility for \eqref{eq:Pan} to hold, we on the contrary use their argument to suggest that $\glm$ does not depend continuously on $b$ and $\eps$.  
	\end{rem}
	
	We end this section by a comment about general domains. We do believe that the analogues of \eqref{eq:Pan} and \eqref{eq:GL degree result} continue to hold for any smooth $\Om \subset \R ^2$. The estimates of Theorem \ref{theo:disc energy} clearly fail however since the curvature of the domain is not constant in this case. We still think that our method can ultimately allow to generalize Theorems \ref{theo:Pan} and \ref{teo:circulation} to general smooth domains, but important new ingredients are also required. This will be the subject of a future work.

\subsection{Heuristic considerations and sketch of the method}\label{sec:sketch}

The starting point of our analysis is a reduction of the problem to an effective one posed in a small boundary layer and a mapping to boundary coordinates. This involves a clever choice of gauge, decay estimates for the GL order parameter (mostly the so-called Agmon estimates) and a replacement of the induced vector potential by the external one. This part of the analysis is borrowed from the aforementioned references and we have virtually nothing new to say about it. The main steps will be recalled later for the convenience of the reader, an exhaustive reference on these methods being \cite{FH}. In this section we outline the main steps of our approach beyond these classical reductions, starting from a functional that is commonly used in the literature as an intermediate between the full GL energy  and the 1D functional \eqref{eq:intro 1D func}. This will illustrate the arguments we shall use in the proofs of our main results, in particular give a pretty complete picture of the main ingredients for the 
results of Subsection \ref{sec:main generic}. The estimates leading to the results of Subsection \ref{sec:main disc} are based on the same kind of considerations, applied to more complicated functionals however and with significant additional technical difficulties.

\medskip

After the reductions we have just mentioned, the leading order of the functional is given by a model on a half-plane, with a fixed vector potential parallel to the boundary\footnote{The subscript $\rm{hp}$ stands for ``half-plane''. Sometimes this is referred to as a functional on a half-cylinder.}:
\begin{equation}\label{eq:cylind func}
\Ehp [\psi] = \int_{-L} ^L \diff s \int_{0} ^{+\infty} \diff t \lf\{ \left|\left( \nabla - i t \es \right) \psi\right| ^2  + \frac{1}{b} |\psi| ^4 - \frac{2}{b} |\psi| ^2 \ri\}.  
\end{equation}
Here the coordinate $t$ corresponds to the direction normal to the boundary of the original sample, $s$ to the tangential coordinate. Length units have been multiplied by $\eps ^{-1}$ so that $2L=\tx \frac{|\dd \Om|}{\eps}$. 
Note that the only large parameter of this functional is the length $L$. The parameter $b$ is the original one and thus satisfies \eqref{eq:b condition}. Since the domain 
$$D_L := [-L,L] \times \R ^+$$
corresponds to the unfolded boundary layer (with neglected curvature), one must impose periodicity of $\psi$ in the $s$-direction. This is not so important when $L$ is large, and in this section we shall only assume periodicity of $|\psi|$ in the $s$ variable. We denote $\Ehpe (L)$ the minimum of $\Ehp$ under this condition and recall that the constant $E_b$ obtained by X.B. Pan \cite{Pan} in \eqref{eq:FH energy asympt} is in fact given by the large $L$ limit of $ (2L) ^{-1} \Ehpe (L)$.

\medskip

We will argue that $\Ehpe = 2L \eoneo$ where $\eoneo$ is defined at the beginning of Subsection \ref{sec:main generic}. The upper bound $\Ehpe \leq 2L \eoneo$ is immediate, using $\psi (s,t) = \fO(t) e^{-i\alO s}$ as a trial state. The strategy we borrow from \cite{CPRY2,CPRY3} yields the lower bound $\Ehpe \geq 2 L \eoneo$. The main steps are as follows:
\begin{enumerate}
\item When $b<\theo ^{-1}$, $\fO$ is strictly positive everywhere in $\R ^+$. To any $\psi$ we may thus associate a $v$ by setting
\begin{equation}\label{eq:intro dec func}
\psi (s,t) = \fO (t) e^{-i\alO s} v(s,t). 
\end{equation}
Using the variational equation for $\fO$ it is not difficult to see that 
\begin{align}\label{eq:intro dec ener}
\Ehp [\psi] &=  2 L \eoneo + \E_0 [v],
\\\label{eq:intro energy E0}
\E_0 [v] &=  \int_{D_L} \diff s \diff t \: \fO ^2(t) \: \bigg\{ \lf| \nabla v \ri|^2 - 2 (t+\alO)\es \cdot \jv + \frac{1}{2 \hex} \fO^2(t) \lf( 1 - |v|^2 \ri)^2 \bigg\},
\end{align}
with $\es$ the unit vector in the $s$-direction and
\[
 \jv = \tx\frac{i}{2} \left( v \nabla v ^* - v ^* \nabla v\right)
\]
the \emph{superconducting current} associated with $v$. This kind of energy decoupling, originating in \cite{LM}, has been used repeatedly in the literature (see \cite{CR,CRY,R2} and references therein).
\item In view of \eqref{eq:intro dec ener}, we need to prove a lower bound to the reduced functional $\E_0$. Here our strategy differs markedly from the spectral approach of \cite{AH,FHP}. We first note that the field $2(t+\alO)\fO ^2 (t)\es$ is divergence-free and may thus be written as $\nabla ^{\perp} F_0$ with a certain $F_0$, that we can clearly choose as 
\[
F_0 (t,s) = F_0 (t) = 2 \int_{0} ^{t} \diff \eta \: (\eta + \alO)\fO ^2 (\eta), 
\]
by fixing $F_0 (0) = 0$. Then it follows from the definition of $\alO$ and $\fO$ that $ F_0 (+\infty) = 0$ (this is the Feynman-Hellmann principle applied to the functional $\foneOal$). Using Stokes' formula on the term involving $\jv$ in \eqref{eq:intro energy E0} we thus have 
\[
\E_0 [v] : =  \int_{D_L} \diff s \diff t \: \left\{ \fO ^2(t) \lf| \nabla v \ri|^2 + F_0 (t) \muv + \frac{1}{2 \hex} \fO^4(t) \lf( 1 - |v|^2 \ri)^2 \right\},
\]
where 
\[
\muv = \curl \: \jv
\]
is the \emph{vorticity} associated to $v$. We also use the periodicity of $|\psi|$ here.
\item It is not difficult to prove that $F_0$ is negative, so that we may bound below 
\[
 \E_0 [v] \geq \int_{D_L} \diff s \diff t \:  \left(\fO ^2(t) \lf| \nabla v \ri|^2 + F_0 (t) |\muv|\right).
\]
We now make the simple but crucial observation that the vorticity is locally controlled by the kinetic energy density:
\begin{equation}\label{eq:intro control vortic}
|\muv| \leq |\nabla v| ^2,   
\end{equation}
so that we have 
\begin{equation}\label{eq:intro low bound}
\E_0 [v] \geq \int_{D_L} \diff s \diff t \:  \left(\fO ^2(t) + F_0 (t)\right) \lf| \nabla v \ri|^2. 
\end{equation}
\item In view of \eqref{eq:intro dec ener} and \eqref{eq:intro low bound}, the sought-after lower bound follows if we manage to prove that the cost function
\begin{equation}\label{eq:intro cost}
K_0 (t):= \fO ^2(t) + F_0 (t)\geq 0 
\end{equation}
for any $t\in \R ^+$. This is the crucial step and we prove in Subsection \ref{sec:model half plane} below that this is indeed the case under condition \eqref{eq:b condition} (in fact $b=1$ is also allowed).
\end{enumerate}

Steps 1 to 3 are general and follow \cite{CPRY2,CPRY3}, Step 1 having also been used previously in this context \cite{AH}. It is in Step 4 that the specificities of surface superconductivity physics enter, in the form of the properties of the 1D functional \eqref{eq:intro 1D func}. In particular the proof relies on the effective potential appearing in the 1D functional being quadratic\footnote{A closer inspection shows that the crucial point is that $\dd_\alpha (t+\alpha) ^2 = \dd_t (t+\alpha) ^2$.}, on the optimality of $\fO$ and $\alO$ and of course on the 1D aspect of the problem. A noteworthy point is that our main ingredient, the lower bound to the cost function \eqref{eq:intro cost}, holds in the regime \eqref{eq:b condition} (actually for $1\leq b < \theo^{-1}$) and only there\footnote{In fact for $b>\theo^{-1}$ \eqref{eq:intro cost} is trivially true since $\fO \equiv 0$, but then \eqref{eq:intro dec func} does not make sense, so the method fails altogether.}. The method is thus sharp in 
this respect.

As we mentioned earlier, the main insight is to notice that only the formation of vortices could lower $\E_0[v]$ to make it negative. Indeed the only potentially negative term in \eqref{eq:intro energy E0} is that involving the superconducting current $\jv$, that we may rewrite using the vorticity. Thinking of $v$ in the form $\rho e ^{i\varphi}$, $\jv = \rho ^2 \nabla \varphi$ corresponds to a velocity field, which justifies the name ``vorticity'' for $\muv$ (as in fluid mechanics). A reasonable guess is to suppose that on average $\rho =|v| \sim 1$, which favors the last term of \eqref{eq:energy E0}. Then, heuristically, $\muv$ may be non-trivial only if the phase $\varphi$ has singularities, i.e., vortices. These could a priori lower the energy by sitting where $F_0$ is minimum but inequalities \eqref{eq:intro control vortic} and \eqref{eq:intro cost} together show that the kinetic energy cost of such vortices would always dominate the gain. Note that although we think in terms of vortices for the 
sake of 
heuristics, we do not need to apply any sophisticated vortex ball method as in \cite{CPRY1,CR,CRY,R2} to bound $\E_0$ from below.   

\medskip

The rest of the paper contains the proofs of our main results, for which we will have to present variants of the above strategy to accommodate various technical aspects of the problem. We start in Section \ref{sec:model prob} by analyzing the 1D functionals \eqref{eq:intro 1D func} and \eqref{eq:intro 1D func disc}. In particular we define the associated cost functions and prove their positivity properties, which are the main new technical ingredients of the present contribution. Section \ref{sec:proof generic} contains a sketch of the now standard procedure of reducing the GL functional to a small boundary layer and replacing the magnetic vector potential. We also conclude there the proof of our Theorem \ref{theo:generic} about general domains. In Section \ref{sec:proof disc} we restrict the setting to disc samples and prove the results of Subsection \ref{sec:main disc}.  An Appendix gathers technical estimates that we use here and there.

\section{Analysis of Effective One-dimensional Problems}\label{sec:model prob}

As will be discussed in more details below, once the usual standard reduction to the boundary layer and appropriate change of gauge have been performed, we are left with the following functional 
\begin{multline}\label{eq:GL func bound}
	\annfpre[\psi] : = \int_{\ann} \diff s \diff t \: \lf(1 - \eps \curv t \ri) \lf\{ \lf| \partial_t \psi \ri|^2 + \frac{1}{(1- \eps \curv t)^2} \lf| \lf( \partial_s + i \aae(s,t) \ri) \psi \ri|^2 \ri.	\\	
	\lf. - \frac{1}{2 \hex} \lf[ 2|\psi|^2 - |\psi|^4 \ri]  \ri\},
\end{multline}
where we have set
\beq
	\aae(s,t) : =- t + \half \eps \curv t^2 + \eps \deps , 	
\eeq
and
\beq\label{eq:deps}
	\deps : = \frac{\gamma_0}{\eps^2} - \lf\lfloor \frac{\gamma_0}{\eps^2} \ri\rfloor,	\qquad	 \gamma_0 : = \frac{1}{|\partial \Omega|} \int_{\Omega} \diff \rv \: \curl \aavm,
\eeq
$ \lf\lfloor \: \cdot \: \ri\rfloor $ standing for the integer part. The coordinates $\eps s$ and $\eps t$ correspond respectively to the tangential and normal coordinates in a tubular neighborhood of $\dd \Om$, $$ \ann:= \lf\{ (s,t) \in \left[0, \tx\frac{|\partial \Omega|}{\eps} \right] \times \left[0,c_0 |\log\eps|\right] \right\} $$
with $c_0$ a constant independent of $\eps$. The function $k(s)$ is the curvature of $\dd \Om$ as a function of $\eps s$. The connection between \eqref{eq:gl func} and \eqref{eq:GL func bound} will be investigated in Section \ref{sec:proof generic}. In the present section we are concerned with the study of effective functionals obtained when plugging some physically relevant ans\"atze in \eqref{eq:GL func bound} and/or neglecting some lower order terms. The analysis of these model functionals provide the main technical ingredients of the proofs of our main results.

For general domains we first neglect the terms involving the curvature $k(s)$ in \eqref{eq:GL func bound}. Indeed, they all come multiplied by an $\eps$ factor and will thus contribute only to the subleading order of the energy. Setting artificially $k(s)\equiv 0$ (that is, approximating the original domain by a half-plane), making the ansatz 
\begin{equation}\label{eq:ansatz}
\psi (s,t) = f(t) e^{- i  \left( \alpha  + \eps\deps  \right) s}  
\end{equation}
and integrating over the $s$ variable, we obtain the 1D functional (times $ \frac{|\partial \Omega|}{\eps} $)
\beq\label{eq:1D func half plane}
\fone_{0,\alpha}[f] : = \int_0^{+\infty} \diff t \lf\{ \lf| \partial_t f \ri|^2 + (t + \alpha )^2 f^2 - \tx\frac{1}{2b} \lf(2 f^2 - f^4 \ri) \ri\},
\eeq
which is known to play a crucial role in the surface superconducting regime \cite{AH,FH,FHP}. We have also set $\eps = 0$ so that the integration domain is now the full half-line.

The improvement of our results when the domain is a disc comes from the simple observation that in this case it is not necessary to drop the curvature terms to obtain a 1D functional out of the ansatz \eqref{eq:ansatz}. Indeed, taking the curvature to be a constant $k(s)\equiv k$ and plugging \eqref{eq:ansatz} in \eqref{eq:GL func bound}, we obtain 
\begin{equation}
\label{eq:1D func disc}
\fone_{k,\alpha}[f] : = \int_0^{c_0|\log\eps|} \diff t \: (1-\eps k t )\lf\{ \lf| \partial_t f \ri|^2 + \frac{(t + \alpha - \frac12 \eps k t ^2 )^2}{(1-\eps k t ) ^2} f^2 - \frac{1}{2b} \lf(2 f^2 - f^4 \ri) \ri\},
\end{equation}
after integration over the $s$ variable and extraction of the factor $|\partial \Omega|/\eps$. This functional includes subleading corrections of order $\eps$ but retains the 1D character of \eqref{eq:1D func half plane} (which is just \eqref{eq:1D func disc} with $k=0$) and is thus amenable to a similar, although technically more involved, treatment.

For the refinements in the disc case it is important that we retain the definition of \eqref{eq:1D func disc} on an interval and not the full half-line, since the functional makes sense only for $ \eps k t < 1 $. In addition the tuning of the value of $c_0$ is going to play a role in the proof. The differences in the treatment of the problems on the half-line and on a large interval are not very important until we introduce the cost functions in Subsections \ref{sec:model half plane} and \ref{sec:model disc}, so we only make remarks in this direction. The exponential decay \eqref{fal point l u b} of $ \fkal $ however guarantees that once $ k $ is set equal to $ 0 $ in \eqref{eq:1D func disc}, all the discrepancies with \eqref{eq:1D func half plane} are just of order $ \OO(\eps^{\infty}) $. 
it is easy to see that, with $c_0$ large enough, one may freely use one or the other convention. 

\medskip

We now set some notation that is going to be used in the rest of the paper. We denote by $\fkal$ the minimizer of \eqref{eq:1D func disc}, with $\eonekal$ the corresponding ground state energy, i.e.,
\beq
	\label{eq:eonekal}
	\eonekal : = \inf_{f \in H^1(\ie)} \fone_{k,\alpha}[f] = \fone_{k,\alpha}[\fkal],
\eeq
with
\beq
	\label{teps}
	 \ieps : = [0,\teps], \qquad \teps : = c_0 |\log\eps|.
\eeq
We will optimize $ \eonekal $ with respect to $\al\in\R$, thereby obtaining an optimal 1D energy that we denote
\beq
	\label{eq:eonek}
	\eonek : = \min_{\al \in \R} \eonekal.
\eeq
Any minimizing $\alpha$ will be denoted $ \alk $ and the minimizer of $\fone_{k,\alk} := \fonek$, achieving $\eone_{k,\alk} := \eonek$ will be written $\fk$ for short. We use the same conventions for the half-plane case, simply setting $k=0$. For shortness we shall also denote by 
\beq
	\label{pot}
	\potkal( t) : = \bigg( \frac{ t  + \alpha - \half \eps k  t^2}{1 - \eps k t} \bigg)^2,
\eeq
the potential appearing in \eqref{eq:1D func disc}.
Note that the potential $ \potkal $ is a translated harmonic potential in $\ieps$ up to corrections of order $ \eps |\log\eps|^3 $: 
\beq
	\label{potkal est}
	\potkal(t) = \lf(  t + \alpha \ri)^2 + \OO(\eps|\log\eps|^3)= \potOal (t) + \OO(\eps|\log\eps|^3).
\eeq
As before we shall drop the $\alpha$ subscript when $\alpha$ is taken to be $\alk$, obtaining the potentials $\potk$ and~$\potO$.

\medskip

In this section we study the above functionals and in particular prove the desired positivity properties of the associated cost functions (defined below). We start in Subsection \ref{sec:disc half plane} by providing elementary properties common to both functionals. The results for the functional \eqref{eq:1D func half plane} have already been proved elsewhere (see \cite[Section 14.2]{FH} and references therein) and we recover them as a particular case ($k=0$, $\eps = 0$) of our study of \eqref{eq:1D func disc} for arbitrary $k$. We proceed to discuss the positivity of the cost function associated to the half-plane functional \eqref{eq:1D func half plane}  and the disc functional \eqref{eq:1D func disc} in Subsection \ref{sec:model half plane} and \ref{sec:model disc} respectively. Again, the half-plane case is actually contained in the disc case, but we provide a simpler proof containing the main ideas when $k= 0$. The proof of Theorem \ref{theo:generic} uses only this case and the refined analysis of 
Subsection \ref{sec:model disc} will be used only in Section~\ref{sec:proof disc}.

\subsection{Preliminary analysis of the effective functionals}\label{sec:disc half plane}

We start by discussing elementary properties of the minimization of $ \fonekal $:

	\begin{pro}[\textbf{Minimization of $ \fonekal $}]
		\label{min fone: pro}
		\mbox{}	\\
		For any given $ \alpha \in \R $, $ k \geq 0 $ and $ \eps$  small enough, there exists a minimizer $ \fkal $ to $ \fonekal$, unique up to sign, which we choose to be non-negative. It solves the variational equation
		\beq
			\label{var eq fal}
			- \fkal^{\prime\prime} + \tx\frac{\eps k}{1 - \eps k t} \fkal^{\prime} + \potkal ( t) \fkal = \tx\frac{1}{\hex} \lf(1 - \fkal^2 \ri) \fkal
		\eeq
		with boundary conditions $ \fkal^{\prime}(0) =  \fkal^{\prime}(c_0|\log\eps|) = 0 $. Moreover $ \fkal $ satisfies the estimate
		\beq
			\label{fal estimate}
			\lf\| \fkal \ri\|_{L^{\infty}(I_{\eps})} \leq 1
		\eeq
		and it is monotonically decreasing for $  t \geq \max \lf[ 0, -\alpha + \tx\frac{1}{\sqrt{\hex}} -C \eps \ri] $. 
		In addition $ \eonekal $ is a smooth function of $ \alpha \in \R $ and
		\beq
			\label{eone explicit}
			\eonekal = - \frac{1}{2\hex} \int_{I_{\eps}} \diff  t \: (1 - \eps k  t) \fkal^4( t).
		\eeq
	\end{pro}

	\begin{rem}(Half-plane case $ k = 0 $) \\
		Strictly speaking in the case $ k = 0 $, $\eps = 0$ there is no boundary condition at $ \infty $, so the only condition satisfied by $ \fOal $ is $ \fOal^{\prime}(0) = 0 $. In addition the variational equation \eqref{var eq fal} turns into
		\begin{equation}\label{eq:1D half plane vareq}
			- \fOal^{\prime\prime} + (t + \al)^2 \fOal = \tx\frac{1}{\hex} \lf(1 - \fOal^2 \ri) \fOal, 
		\end{equation}
		which is familiar from preceding works on this problem (see in particular \cite{FHP} for a refined analysis).
	\end{rem}

	\begin{proof}
		Uniqueness and non-negativity of the minimizer are straightforward consequences of the strict convexity of the functional in $ \rho : = f^2 $. Note however that $ \fkal \equiv 0 $ is in fact the minimizer in a certain range of the parameters (see below). The variational equation for $ \fkal $ is the Euler-Lagrange equation for the functional $ \fonekal $ and the boundary conditions simply follow from the choice of the minimization domain $ H^1(I_{\eps}) $.
	
		The upper bound on $ \sup_{I_{\eps}} \fkal $ is a simple consequence of the maximum principle (see, e.g., \cite[Proof of Proposition 10.3.1]{FH}), while the monotonicity for $  t $ larger than $ - \alpha + \frac{1}{\sqrt{b}} -C \eps $ (assuming that the expression is positive) can be easily obtained by integrating the variational equation \eqref{var eq fal} in $ [t,c_0|\log\eps|] $, which yields
		\beq
			(1 - \eps k t) \fkal^{\prime}( t) = \int_{ t}^{c_0 |\log\eps|} \diff \eta \: (1 - \eps k \eta) \lf[  \tx\frac{1}{\hex} (1 - \fkal^2) - \potkal(\eta) \ri] \fkal,
		\eeq
		thanks to Neumann boundary conditions. For $  t \geq -\alpha + \tx\frac{1}{\sqrt{\hex}} -C \eps $, $ \potkal(t) \geq \frac{1}{\hex} $ so that the integrand and therefore the whole expression are negative.
	\end{proof}

As already mentioned in the proof of Proposition \ref{min fone: pro} above, the minimizer $ \fkal $ can as well be identically zero: all the properties \eqref{var eq fal}-\eqref{fal estimate} are indeed satisfied by the trivial function $ f \equiv 0 $, in which case the energy itself would be $ \eonekal = 0 $ by \eqref{eone explicit}. We thus have to single out the proper conditions on $ \alpha $ and $ \hex $ guaranteeing that the minimizer $ \fkal $ is non-trivial, which are of crucial importance since the regime $\fO \equiv 0$ corresponds to applied magnetic fields above $\Hccc$. 

We introduce a linear operator associated with $ \fonekal $, whose spectral analysis allows to investigate the existence of a non-trivial minimizer $ \fkal $: we denote by $ \mue $ the ground state energy of the Schr\"{o}dinger operator
\beq\label{eq:Hkal}
	\hkal : = - \partial_t^2 - \tx\frac{\eps k}{1 - \eps k t} \partial_t + \potkal ( t),
\eeq
with Neumann boundary conditions in $ \mathscr{H} : =  L^2(I_{\eps}, (1 - \eps k  t) \diff  t) $, i.e.,
\beq
	\mue : = \inf_{u \in L^2(I_{\eps}), \lf\| u \ri\|_{ \mathscr{H}} = 1} \bra{u} \hkal \ket{u}_{\mathscr{H}}.
\eeq
Several properties of this Schr\"{o}dinger operator are collected in the Appendix (see Section \ref{1d So: sec}).

The main result about the non-triviality of $ \fal $, which is the analogue of \cite[Proposition 14.2.2]{FH}, is the following:

	\begin{pro}[\textbf{Existence of a non-trivial minimizer $ \fkal $}]
		\label{nontrivial: pro}
		\mbox{}	\\
		For any given $ \eps$  small enough, $ \alpha \in \R $ and $ k \geq 0 $, the minimizer $ \fkal $ is non-trivial, i.e., $ \fkal \neq 0 $, if and only if 
		\beq
			\hex^{-1} > \mue.
		\eeq
		In particular if the condition is satisfied, $ \fkal $ is the unique positive ground state (not normalized) of the Schr\"{o}dinger operator 
		\beq
			\hkal - \tx\frac{1}{\hex} \lf(1 - \fkal^2\ri)
		\eeq
		 with Neumann boundary conditions.
	\end{pro}

	\begin{proof}
		The result can be proved exactly as in \cite[Proposition 14.2.2]{FH}. Indeed using the same cut-off function $ \chi_N $, one can conclude that $ \fkal \equiv 0 $, if $ \hex^{-1} \leq \mue $. The opposite implication follows from the existence of a positive normalized ground state $ \phikal $ of $ \hkal $: picking $ a \phikal $ for some sufficiently small $ a $ as a trial state, we obtain
		\bml{
			\fonekal[a \phikal] = \bra{a \phikal} \hkal  \ket{a \phikal} - \tx\frac{a^2}{\hex} + \frac{a^4}{2\hex} \lf\| \phikal \ri\|_4^4	\\
			 \leq a^2 \lf[ \mue - \tx\frac{1}{\hex} + \frac{a^2}{2\hex} \lf\| \phikal \ri\|_4^4 \ri] < 0 = \fonekal [0],
		}
		because $\mue - \tx\frac{1}{\hex} < 0$ by assumption. This rules out the possibility of the minimizer being identically zero.
	\end{proof}

	We also give a criterion in terms of  $ \muosc $, the lowest eigenvalue of the $\alpha$-shifted 1D harmonic oscillator on the half-line with Neumann boundary conditions: 
	\beq
		\label{eq:hosc}
		\hosc_{\alpha} : = - \partial_{ t}^2 + ( t + \alpha)^2,
	\eeq
	acting on $ L^2(\R^+,\diff  t) $. This criterion is more convenient because it does not depend on $\eps$. It follows from a comparison between $\mue$ and $\muosc$ that we provide in the Appendix, Section \ref{1d So: sec}.

	\begin{cor}[\textbf{Existence of a non-trivial minimizer $ \fkal $, continued}]
		\label{nontrivial: cor}
		\mbox{}	\\
		Let $ 1 < \hex < \theo^{-1} $ and $ \aldown_i(k,\hex), \alup_i(k,\hex) $, $ i = 1,2 $, be defined as in \eqref{al up down}. Then for any $ \al \in (\alup_2,\aldown_1) $ the minimizer $ \fal $ is non-trivial. In the case $k=0$ the minimizer is non trivial if and only if $\hex^{-1}  > \muosc$.
	\end{cor}

	\begin{proof}
		The result is obtained by a direct combination of Corollary \ref{mue min: cor} (see in particular \eqref{cond hex}) with the above Proposition \ref{nontrivial: pro}. The part about the $k=0$ case is contained in \cite[Proposition 14.2]{FH}.
	\end{proof}

We can now show the existence of an optimal phase $ \alk $ minimizing $ \eonekal $ with respect to $ \alpha \in \R $, for any $ \hex \in (1,\theo^{-1}) $:

	\begin{lem}[\textbf{Optimal phase $ \alk $}]
		\label{lem:opt phase}
		\mbox{}	\\
		For any $ 1 < \hex < \theo^{-1} $, $ k \geq 0 $ and $ \ep$ small enough, there exists at least one  $ \alk $ minimizing $ \eonekal $:
		\beq
			\label{eq:optimal energy}
			\inf_{\alpha \in \R} \eonekal = \eone_{k,\alk} =: \eonek.
		\eeq
		Setting $ \fk : = f_{k,\alk} $ we have that $ \fk > 0$ everywhere and 
		\beq
			\label{FH nonlinear}
			\int_{I_{\eps}} \diff  t \: \frac{ t + \alk - \half \eps k t^2}{1 - \eps k t} \fk^2( t) = 0.
		\eeq
	\end{lem}

	\begin{proof}
		The existence of a minimizer is basically a consequence of Corollary \ref{mue min: cor}: for any $ \hex \in (1,\theo^{-1}) $, one can find four negative values $ \aldown_i(k,\hex) $ and $ \alup_i(k,\hex) $, $ i = 1,2 $, such that (recall that $ \alup_1 < \aldown_2 $)
		\beqn
			\eonekal & = & 0,		\quad	\mbox{for } \al \in (-\infty,\aldown_1] \mbox{ or } \al \in [\alup_2,\infty),	\nonumber	\\\
			\eonekal & < & 0,		\quad	\mbox{for } \al \in (\alup_1,\aldown_2).	\nonumber
		\eeqn
		Obviously this implies the existence of a minimizer $ \alk $. In addition $ \eonekal $ is a smooth function of $ \alpha $ in the interval $ (\alup_1,\aldown_2) $ and studying its derivative with respect to $\al$ yields  (Feynman-Hellmann principle) 
		\beq
			\label{FH alpha}	
			\partial_{\alpha} \eonekal = 2 \int_{I_{\eps}} \diff  t \: \frac{ t + \al - \half \eps k t^2}{1 - \eps k t} \fal^2( t),
		\eeq  
		so that at $ \alk $ \eqref{FH nonlinear} must hold true.
	\end{proof}

In the rest of the paper we are going to use several times, in particular when estimating the cost function, the following pointwise bounds, whose somewhat technical proofs are discussed in the Appendix (Section \ref{sec:app est}) together with other useful estimates. A similar result appeared in \cite[Theorem 3.1]{FHP}, although the constants involved in the estimate there are not uniform in $ \eps $, unlike those involved in the bounds below. Note that the pointwise estimates are formulated only for $ \hex \in (1,\theo^{-1}) $ and $ \al \in (\alup_2,\aldown_1) $ defined in \eqref{al up down}, where we already know that the minimizer is not identically zero.

	\begin{pro}[\textbf{Pointwise estimates for $ \fal $}]
		\label{pro:point est fal}
		\mbox{}	\\
		For any $ 1 < \hex < \theo^{-1} $, $ \al \in (\alup_2,\aldown_1) $, $ k \geq 0 $ and $ \eps \ll 1 $, there exist two positive constants $ c, C > 0 $ independent of $ \eps $ such that
		\beq
			\label{fal point l u b}
			c \: \exp\Big\{ - \tx\frac{1}{2}\big(  t + \sqrt{2} \big)^2 \Big\} \leq \fal( t) \leq C \: \exp\lf\{ - \tx\frac{1}{2} \lf(  t + \al \ri)^2 \ri\},
		\eeq
		for any $  t \in \ie $.
	\end{pro}

We end this section by introducing the potential function associated with $\fk$:
\beq
\label{Fk}
\boxed{\Fk (t) : = 2 \int_0^t \diff \eta \: (1 - \eps k \eta) \fk^2(\eta) \frac{\eta + \alk - \half \eps k \eta^2}{(1 - \eps k \eta)^2}.}
\eeq
The motivation for introducing these objects will become clearer in Section \ref{sec:proof disc} (see also the heuristic discussion in Subsection \ref{sec:sketch}). We collect some of their properties in the following lemma. 

\begin{lem}[\textbf{Properties of the potential function $\Fk$}]
		\label{lem:F prop}
		\mbox{}	\\
		For any $ 1 < \hex < \theo^{-1} $, $ k \geq 0  $ and $ \eps$ sufficiently small, let $ \Fk$ be the function defined in \eqref{Fk}. Then we have
		\beq
			\label{F prop}
			\Fk(t) \leq 0,	\quad \mbox{in } \ieps,	\qquad \Fk(0) = \Fk(\teps) = 0.
		\eeq
		In the case $k=\eps=0$, the equation $\Fk(\teps) = 0$ should be read as $ \lim_{t\to \infty} F_0 (t) = 0 $.
	
	\end{lem}

	\begin{proof}
		We observe that $ \Fk(\teps) = 0 $ is simply \eqref{FH nonlinear}, the first order condition for $ \alk $ being a minimizer of $ \eonekal $. The fact that $\Fk(0) = 0$ immediately follows from the definition. 

		On the other hand 
		\beq
			\Fk^{\prime}( t) =  2 \frac{t + \alk - \half \eps k t^2}{1 - \eps k t} \fk^2(t),
		\eeq
		and thanks to the negativity of $ \alk $ and positivity of $ \fk $, we obtain that $ \Fk^{\prime}( t) \leq 0 $ in a neighborhood of the origin. Moreover $ \Fk^{\prime} $ can vanish (again by strict positivity of $ \fk $) only at a single point $  t_k $ where $ t_k + \alk - \half \eps k t_k^2 = 0 $, i.e.,
		\bdm
			t_k = |\alk| + \OO(\eps),
		\edm
		which has then to be a minimum point for $ \Fk $. For $  t >  t_k $, $ \Fk( t) $ is increasing but since $  \Fk(\teps) = 0 $, it also remains negative for any $  t \in \ie $.
	\end{proof}

\subsection{The cost function in the half-plane case}\label{sec:model half plane}

The cost function that will naturally appear in our investigation of \eqref{eq:GL func bound} for general domains is
	\beq
		\label{K0}
		\boxed{\KO(t) : = \fO^2(t) + \FO(t) =  \fO^2(t) + 2 \int_0^{ t} \diff \eta \: \left(\eta + \alO \right)\fO^2(\eta).}
	\eeq

The result we aim at is the following: 

\begin{pro}[\textbf{Positivity of the cost function in the half-plane case}]\label{pro:K0 positive}\mbox{}\\
Let $ \KO(t) $ be the function defined in \eqref{K0}. For any $ 1 \leq \hex < \theo^{-1} $, it holds
\beq
\label{eq:K0 positive}
\KO(t) \geq 0, \quad \mbox{ for any } t \in \R ^+.
\eeq
\end{pro}

\begin{rem}(Extreme regime $ b \to \theo^{-1} $)	\\
A simpler computation can be performed in the case where $b$ converges to $\theo ^{-1}$. In this case it is known (see \cite{FHP} for further details) that $\fO$ is approximately proportional to the first eigenfunction of the harmonic oscillator \eqref{eq:hosc}, with $\alpha_0 = - \sqrt{\theo} $, the minimizer of the oscillator ground state energy (see \eqref{eq:min muosc}). Replacing $\fO$ by this function and following the steps below, one obtains a similar result with a simpler proof, since the nonlinearity in \eqref{eq:1D half plane vareq} can be neglected. This can give an idea of the mechanism at work, but is certainly not sufficient for the proof of our main results in the whole regime \eqref{eq:b condition}.

\end{rem}

The first step in the proof of Proposition \ref{pro:K0 positive} is an alternative expression for the potential function~$\FO$:

\begin{lem}[\textbf{Alternative expression of $\FO$}]\label{lem:F0 alternate}\mbox{}\\
For any $t\in \R ^+$, it holds
\begin{equation}\label{eq:F0 alternate}
\FO (t) = - {\fO '}^2 (t)  + \left( t + \alO \right) ^2 \fO^2 (t)  - \frac{1}{\hex} \fO ^2 (t) + \frac{1}{2\hex} \fO^4 (t) . 
\end{equation}
\end{lem}


\begin{proof}
We first write
\[
2 \int_0^{ t} \diff \eta \: \left(\eta + \alO \right)\fO^2(\eta) = \int_0 ^t \diff \eta \: \fO^2(\eta) \: \dd_\eta \left( \eta + \alO\right)  ^2
\] 
and integrate by parts to obtain 
\[
\FO (t) = - \alO ^2 \fO (0) ^2 + \left( t + \alO\right) ^2 \fO (t) ^2 - 2 \int_0 ^t \diff \eta \: \left( \eta + \alO\right)  ^2 \fO (\eta) \fO '(\eta),
\]
which turns into
\[
\FO (t) = - \alO ^2 \fO ^2 (0) + \left( t + \alO\right) ^2 \fO ^2 (t) - 2 \int_0 ^t \diff \eta \: \left( \fO ^{\prime \prime} (\eta) + \frac{1}{\hex} \fO (\eta) - \frac{1}{\hex} \fO  ^3 (\eta)\right)   \fO '(\eta) 
\]
after using the variational equation \eqref{eq:1D half plane vareq} to replace the $\left( \eta + \alO\right)  ^2 \fO (\eta)$ term in the integral. We then use the trivial identities 
\begin{equation*}
\int_0 ^t \diff \eta \: f f' = \frac12 \left[ f ^2 \right]_0 ^t, \quad \int_0 ^t \diff \eta \: f ^{\prime \prime} f' = \frac12 \left[ (f') ^2 \right]_0 ^t,  \quad \int_0 ^t \diff \eta \: f ^3 f' = \frac{1}{4} \left[ f ^4 \right]_0 ^t  
\end{equation*}
and deduce (using the Neumann boundary condition satisfied by $\fO$ at the origin)
\[
\FO (t) = - {\fO '}^2 (t)  + \left( t + \alO \right) ^2 \fO^2 (t)  - \tx\frac{1}{\hex} \fO ^2 (t) + \frac{1}{2\hex} \fO ^4 (t) - \alO ^2 \fO ^2 (0) + \frac{1}{\hex} \fO ^2 (0) - \frac{1}{2\hex} \fO ^4 (0).
\]
To obtain the final formula, we recall that (this is \eqref{FH nonlinear} for $k = 0$ and $\eps = 0$)
\[
F_0(+\infty) = 2 \int_0^{ +\infty } \diff \eta \: \left(\eta + \alO \right)\fO^2(\eta) = 0,
\]
so that, by the decay of $ \fO $ and $ \fO' $ at $ + \infty $, which can be obtained combining the pointwise estimates of Proposition \ref{pro:point est fal} with \eqref{eq:fOal derivative}, we deduce
\[
- \alO ^2 \fO^2 (0)  + \tx\frac{1}{\hex} \fO ^2 (0) - \frac{1}{2\hex} \fO^4 (0) = \disp\lim_{t\to + \infty} \left[ - (t + \alO) ^2 \fO ^2 (t)  + \tx\frac{1}{\hex} \fO ^2 (t) - \frac{1}{2\hex} \fO ^4 (t) \right]= 0,
\]
and the proof is complete.
\end{proof}

We now complete the

\begin{proof}[Proof of Proposition \ref{pro:K0 positive}]
%
Lemma \ref{lem:F0 alternate} tells us that
\begin{equation}\label{eq:K0 alternate}
\KO (t) = \left( 1 - \frac{1}{\hex} \right) \fO  ^2(t) - {\fO '} ^2 (t) + \left( t + \alO \right) ^2 \fO ^2 (t) + \frac{1}{2\hex} \fO ^4  (t).
\end{equation}
Using the Neumann boundary condition, the decay of $\fO$ and its derivative and the assumption $b\geq1$, we have 
\[
\KO (0) =  \left(1- \frac{1}{\hex} \right) \fO ^2 (0) + \alO ^2 \fO (0) ^2 +  \frac{1}{2\hex} \fO^4 (0)  \geq 0
\]
and 
$$\lim_{t\to +\infty} \KO (t)= 0,$$
so if $\KO$ became negative somewhere in $\R ^+$, it should have a global minimum at some point $ t_0 > 0 $. Let us then compute the derivative of $\KO$: by the definition \eqref{K0}
\[
\KO'(t) = 2  \fO (t) \fO ' (t) + 2 \left( t + \alO\right) \fO ^2 (t),
\]
so that at any critical point $ t_0 $ of $\KO$ we must have
\[
\fO '(t_0) = - \left( t + \alO\right) \fO (t_0)  
\]
because $\fO$ is strictly positive. Plugging this into \eqref{eq:K0 alternate} we find that at any critical point $ t_0 $ of $\KO$, it also holds 
\[
\KO (t_0) =  \left( 1 - \frac{1}{\hex} \right) \fO ^2(t_0) + \frac{1}{2\hex} \fO ^4 (t_0), 
\]
which is clearly positive when $\hex \geq 1$. We thus conclude that the minimum of $\KO$ must be positive, which ends the proof.
\end{proof}

\subsection{The cost function in the disc case}\label{sec:model disc}

We now investigate the properties of the cost function associated with \eqref{eq:1D func disc}. The argument is essentially a perturbation of the one we gave before for the case $k=0$, but, due to the presence of a non-zero curvature $k$, the positivity property we are after is harder to prove. Actually we are only able to prove the desired result in some subregion of $\ieps$ given by
	\beq
		\label{annb}
		\annbk : = \lf\{ t \in \ieps \: : \fk\lf(t\ri) \geq |\log\eps|^3 \fk(\teps) \ri\},
	\eeq
	which is an interval in the $  t $ variable, i.e.,
	\beq
		\label{annb bis}
		\annbk = [0, \btik],
	\eeq
	for some $ \btik  \leq \teps $. Indeed, exploiting the upper bound \eqref{fal point l u b} on $ \fk $, it can easily be seen that
	\beq
		\fk(\teps) = \OO(\eps^{\infty}),
	\eeq
	so that the equality $  \fk( t) = |\log\eps|^3 \fk(\teps) $ can only be satisfied by some $  t \gg 1 $, i.e., by Proposition~\ref{min fone: pro}, in the region where $ \fk $ is monotonically decreasing. Therefore $ \btik $ is unique and the lower bound \eqref{fal point l u b} also implies
	\bdm
 		c \exp \lf\{ - \big( \btik + \sqrt{2} \big)^2 \ri\} \leq C |\log\eps|^3 \exp \lf\{ - \lf(\teps+ \alk \ri)^2 \ri\},
	\edm
	which yields
	\beq
		\label{bxi}
		\btik \geq \teps - C \log|\log\eps|  = c_0 |\log\eps| \lf( 1 - \OO\lf( \tx\frac{\log|\log\eps|}{|\log\eps|} \ri) \ri).
	\eeq
	For further convenience we also introduce the constant $ \keps $ explicitly given by
	\beq
		\label{keps}
		\keps : = \lf[ \potk(\teps) - \tx\frac{1}{\hex} \lf(1 - \fk^2(\teps) \ri) \ri] \fk^2(\teps).
	\eeq
	We do not emphasize its dependence on $k$ in the notation. Note that by the decay \eqref{fal point l u b}, we immediately have 
	$$ 0 \leq \keps = \OO(\eps^{\infty}) .$$

	Now we introduce the cost function whose lower estimate is our main ingredient and prove its positivity in $ \annbk $:
	\begin{align}
		\label{K}
		\Kk(t) &= (1 -  \de) \fk^2(t) + \Fk(t)  \nonumber
		\\ &= (1 -  \de)  \fk^2(t) + 2 \int_0^{ t} \diff \eta \: \frac{\eta + \alk - \half \eps k \eta^2}{1 - \eps k \eta} \fk^2(\eta),
	\end{align}
	where $ \de $ is a parameter satisfying
	\beq
		\label{eq:de}
		0 < \de \leq C |\log\eps|^{-4},	\quad	\mbox{as } \eps \to 0,
	\eeq 
	that will be adjusted in the sequel of the paper and help us in proving Theorem \ref{theo:Pan}. To get to the main point, one can simply think of the case $\de=0$, which is sufficient for the energy estimate of Theorem \ref{theo:disc energy}.

	\begin{pro}[\textbf{Positivity of the cost function in the disc case}]
		\label{pro:K positive}
		\mbox{}	\\
		For any $ \de \in \R^+ $ satisfying \eqref{eq:de}, $ 1 < \hex < \theo^{-1} $, $ k > 0 $ and $ \eps$ sufficiently small, let $ \Kk(t) $ be the function defined in \eqref{K} and $\annbk$ the interval \eqref{annb}. Then one has
		\beq
			\label{eq:K positive}
			\Kk(t) \geq 0, \quad \mbox{ for any } t \in \annbk.
		\eeq
	\end{pro}

The proof of this lower bound is rather technical because our main results require that the complement of the region where the positivity  property holds be one where the density $\fk$ is \emph{extremely} small. Indeed, note that using Proposition \ref{pro:point est fal}, $\fk$ is $\OO(\eps ^{\infty})$ outside of $\annbk$, and we do use this fact later in the paper. If we were allowed to work in a region where a stronger lower bound to $\fk$ holds (say $|\log \eps|$ to some negative power), the proof would be essentially identical to that of Proposition \ref{pro:K0 positive}, with some additional naive bounds.  

	\begin{proof}[Proof of Proposition \ref{pro:K positive}]
		We are going to prove that for any $ t \in \ieps $ one has
		\beq
			\label{cost funct lb}
			 \fk^2( t) + 2 \int_0^{ t} \diff \eta \: \frac{\eta + \alk - \half \eps k \eta^2}{1 - \eps k \eta} \fk^2(\eta) \geq |\log\eps|^{-3} \fk^2\lf( t\ri) - \keps.
		\eeq
		Therefore inside $ \annbk $ we obtain
		\bml{
			\Kk(t) \geq |\log\eps|^{-3}\lf( 1 - \de |\log\eps|^3  \ri) \fk^2\lf(t\ri) - \keps \geq |\log\eps|^3 \lf( 1 -  \de |\log\eps|^3 \ri)  \fk^2( \teps) - \keps 	\\
			\geq |\log\eps|^3 \lf( 1 - \OO(|\log\eps|^{-1}) \ri)  \fk^2( \teps) \geq 0.
		}
		In the last line we have made use of the simple bound $ \keps \leq C |\log\eps|^2 \fk^2(\teps) $ that follows directly from the definition \eqref{keps} and Proposition \ref{pro:point est fal}. Note also that by assumption \eqref{eq:de} $ \de |\log\eps|^3 \leq C |\log\eps|^{-1} \ll 1 $.

		In order to prove \eqref{cost funct lb} we set 
		\beq
			 \kt( t) : = (1 - |\log\eps|^{-3}) \fk^2( t) + 2 \int_0^{ t} \diff \eta \: \frac{\eta + \alk - \half \eps \eta^2}{1 - \eps \eta} \fk^2(\eta) + \keps
		\eeq	
		and prove that $ \kt $ is positive in $ \ieps $. For this purpose we note that, by \eqref{F prop},
		\beq
			\label{K boundary values}
			\kt(0) = (1 - |\log\eps|^{-3}) \fk^2(0) + \keps > 0, \qquad \kt(\teps) = (1 - |\log\eps|^{-3}) \fk^2(\teps) + \keps > 0,
		\eeq
		so that, if it becomes negative for some $  t $, it must be $ 0 <  t < \teps$. Let us then compute the derivative of $ \kt $ to find its minimum points:
		\beq
			\label{eq:cost derivative}
			\kt^{\prime}( t) = 2 (1 - |\log\eps|^{-3}) \fk( t) \fk^{\prime}( t) + 2 \frac{ t + \alk - \half \eps k t^2}{1 - \eps k t} \fk^2( t).
		\eeq
		Let $  t_0 $ be a point where $ \kt $ reaches its global minimum and let assume that it is in the interior of $ \ie $, otherwise there is nothing to prove. Then it must be $ \kt^{\prime}( t_0) = 0 $, i.e., thanks to the strict positivity of $ \fk $,
		\beq
			\label{eq:cost condition 1}
			(1 - |\log\eps|^{-3}) \fk^{\prime}( t_0) = - \frac{ t_0 + \alk - \half \eps k t_0^2}{1 - \eps k t_0} \fk( t_0),
		\eeq
		and thus (recall \eqref{pot})
		\beq
			\label{eq:cost condition 2}
			(1 - |\log\eps|^{-3})^2 {\fk^{\prime}}^2  ( t_0) = \potk(t_0)  \fk ^2 ( t_0).
		\eeq
		On the other hand using the identity
		\beq
			\label{pot derivative}
			\potk^{\prime}(t) = 2 \frac{t + \alk - \half k \eps t^2}{1 - \eps k t} + \frac{2 \eps k \potk(t)}{1 - \eps k t},
		\eeq
		in combination with the variational equation \eqref{var eq fal}, we get
		\bml{	
 			\label{cost 3}		
 			2 \int_0^{ t} \diff \eta \: \frac{\eta + \alk - \half \eps k \eta^2}{1 - \eps k \eta} \fk^2(\eta) = \lf[ \potk \fk^2 \ri]_{0}^{ t} - 2 \int_0^ t \diff \eta \: \potk(\eta) \fk(\eta) \fk^{\prime}(\eta) - 2 \eps k \int_0^ t \diff \eta \: \frac{\potk(\eta)}{1 - \eps k \eta} \fk^2(\eta)	\\
			= \lf[ \potk \fk^2  - \tx\frac{1}{2 \hex} \lf( 2 \fk^2 - \fk^4 \ri) \ri]_{0}^{ t} - {\fk^{\prime}}^2( t) + 2 \eps k \int_0^ t \diff \eta \: \frac{1}{1 - \eps k \eta} \lf( {\fk^{\prime}}^2 -  \potk(\eta) \fk^2 \ri).
		} 
		Hence
		\bmln{
 			\kt( t) = (1 - |\log\eps|^{-3}) \fk^2( t) + \lf[ \potk (\eta) \fk^2 - \tx\frac{1}{2 \hex} \lf( 2 \fk^2 - \fk^4 \ri) \ri]_{0}^{ t} - {\fk^{\prime}}^2( t)	\\
			 + 2 \eps k \int_0^ t \diff \eta \: \frac{1}{1 - \eps k \eta} \lf( {\fk^{\prime}}^2 -  \potk(\eta) \fk^2 \ri) + \keps.
		}
		Plugging \eqref{eq:cost condition 1} into the above expression, we obtain
		\bml{
 			\label{eq:min K}
			\min_{ t \in \ie} \kt( t) = \kt( t_0) = \lf[1  - \tx\frac{1}{\hex} + \frac{1}{2\hex} \fk^2( t_0) - |\log\eps|^{-3} + \lf(1 - \tx\frac{1}{(1 - |\log\eps|^{-3})^2} \ri) \potk(t_0) \ri] \fk^2( t_0) 	\\
			- \lf[ \potk(0) - \tx\frac{1}{\hex} + \tx\frac{1}{2\hex} \fk^2(0) \ri] \fk^2(0) + 2 \eps k \int_0^{ t_0} \diff \eta \: \frac{1}{1 - \eps k \eta} \lf[ {\fk^{\prime}}^2 -  \potk(\eta) \fk^2 \ri] + \keps.
		}
		However the l.h.s. of \eqref{cost 3} above vanishes when $  t = \teps$, thanks to \eqref{FH nonlinear}, and thus
		\beq
			\label{boundary cost}
			0 = \lf[ \potk \fk^2  - \tx\frac{1}{2 \hex} \lf( 2 \fk^2 - \fk^4 \ri) \ri]_{0}^{\teps} + 2 \eps k \int_0^{\teps} \diff \eta \: \frac{1}{1 - \eps k \eta} \lf( {\fk^{\prime}}^2 -  \potk(\eta) \fk^2 \ri).
		\eeq
		Such an identity can be used in \eqref{eq:min K}, yielding
		\beq			
 			\label{cost 4}
			\kt( t_0) \geq  \lf[ 1  - \tx\frac{1}{\hex} + \frac{1}{2 \hex} \fk^2( t_0) - C |\log\eps|^{-1} \ri] \fk^2( t_0) + 2 \eps k \int_{ t_0}^{\teps} \diff \eta \: \frac{1}{1 - \eps k \eta} \lf( \potk(\eta) \fk^2 -  {\fk^{\prime}}^2 \ri).
		\eeq

		The first term on the r.h.s. of \eqref{cost 4} is clearly positive thanks to the condition $ \hex >  1 $. Therefore to prove \eqref{cost funct lb} and thus the result, it remains to study the last term of \eqref{cost 4}. We are going to show that
		\beq
			\label{cost 5}
			\int_{ t_0}^{\teps} \diff \eta \: \frac{1}{1 - \eps k \eta} \lf(\potk(\eta) \fk^2 -  {\fk^{\prime}}^2 \ri) \geq - C \fk^2( t_0),
		\eeq
		for some finite constant $ C $. This in turn implies the lower bound \eqref{cost funct lb}: owing to the condition $ \hex >  1 $, we have
		\beq
			\kt( t_0) \geq \lf[ 1  - \tx\frac{1}{\hex} + \frac{1}{2 \hex} \fk^2( t_0) - C\log\eps|^{-1} - C \eps  \ri] \fk^2( t_0)  \geq 0,
		\eeq
		if $ \eps $ is sufficiently small.

		Let us then focus on \eqref{cost 5}: the key ingredient is the analysis of the function
		\beq
			g( t) : = (1 - |\log\eps|^{-3}) \fk^{\prime}( t) + \frac{ t + \alk - \half k \eps  t^2}{1 - \eps k  t} \fk( t) = : (1 - |\log\eps|^{-3}) \fk^{\prime}( t) + A( t) \fk( t),
		\eeq
		which appears in \eqref{eq:cost derivative}, i.e., $ \kt^{\prime}( t) = 2 \fk( t) g( t) $. Computing $ g^{\prime}( t) $ using the variational equation \eqref{var eq fal}, we get
		\bml{
			\label{gprime}
			g^{\prime}( t) = \lf[ \frac{A( t)}{1- |\log\eps|^{-3}}	+ \frac{\eps}{1 - \eps  t} \ri] g( t) + \lf[ 1 - \OO(|\log\eps|^{-3}) \potk(t) - \frac{1 - |\log\eps|^{-3}}{\hex} \lf( 1 - \fk^2( t) \ri) \ri] \fk( t)	\\
			\geq \lf[ \frac{ A( t)}{1- |\log\eps|^{-3}} + \frac{\eps k }{1 - \eps k t} \ri] g( t) + \lf( 1 - \tx\frac{1}{\hex} - \OO(|\log\eps|^{-1}) \ri) \fk( t)	\\
			 > \lf[  \frac{ A( t)}{1- |\log\eps|^{-3}}  + \frac{\eps k}{1 - \eps k t} \ri] g( t),
		}
		by the positivity of $ \fk $, the condition $ \hex > 1 $ and the bound $ \potk(t) \leq \OO(|\log\eps|^2) $ in $\ieps$. On the other hand at the outer boundary
		\beq
			g( \teps) = A( \teps) \fk( \teps) > 0,
		\eeq
		which in particular implies that $  t_0 < \teps$. We now distinguish two cases: if $  t_0 \leq 2 |\alk| $, then the pointwise bounds \eqref{fal point l u b} implies that $ \fk( t_0) > C > 0 $ and therefore 
		 \bdm
			\kt( t_0) \geq C - \OO(\eps|\log\eps|^3) \geq 0
		\edm
		by a simple estimate of the last term in \eqref{cost 4}, so we are done.
		
		From now on we may thus assume $  t_0 \geq 2 |\alk| $. Then $A(t)\geq 0$ for any $t\geq t_0$ and $  t_0 $ is in the region where $ \fk $ is monotonically decreasing (see Proposition \ref{min fone: pro}). We then have $g(t_0) = 0$ by \eqref{eq:cost condition 1} and $g'(t_0) > 0$ by \eqref{gprime}. But again by \eqref{gprime}, as soon as $ g $ gets positive, its derivative becomes positive too. Therefore $g$ remains increasing and positive up to the boundary.

		
		Using this information on $ g $ we estimate
		\bml{
 			\label{cost 6}
 			\int_{ t_0}^{\teps} \diff \eta \frac{1}{1 - \eps k \eta} \lf[ \potk \fk^2 - {\fk^{\prime}}^2 \ri] = \int_{ t_0}^{ \teps} \diff \eta \frac{1}{1 - \eps k \eta} \lf[ A\fk - \fk^{\prime} \ri] \lf[ g + |\log\eps|^{-3} \fk^{\prime} \ri]	\\
			\geq  |\log\eps|^{-3} \int_{ t_0}^{\teps} \diff \eta \frac{1}{1 - \eps k \eta} \lf[ A \fk - \fk^{\prime} \ri] \fk^{\prime} \geq -  |\log\eps|^{-3} \int_{ t_0}^{\teps} \diff \eta \frac{1}{1 - \eps k \eta} \lf[ \potk \fk^2 -A \fk \fk^{\prime} \ri],
		}
		where we have used that, for any $ \eta \geq  t_0 $,
		\bdm
			\fk^{\prime}(\eta) \leq 0,	\qquad  g(\eta) = \fk^{\prime}(\eta) + A(\eta) \fk(\eta) \geq 0.
		\edm
		On the other hand we have the bound
		\bdm
			\int_{ t_0}^{\teps} \diff \eta \frac{1}{1 - \eps k \eta}  \potk \fk^2 \leq C |\log\eps|^3 \fk^2( t_0),
		\edm
		again by monotonicity of $ \fk $ We can then compute
		\bmln{
 			\int_{ t_0}^{\teps} \diff \eta \frac{1}{1 - \eps k \eta} A(\eta) \fk(\eta) \fk^{\prime}(\eta) = \frac{1}{2} \lf[ \frac{A(\eta) \fk^2(\eta)}{1 - \eps k \eta} \ri]^{\teps}_{ t_0} - \frac{1}{2} \int_{ t_0}^{\teps} \diff \eta \frac{1}{1 - \eps k \eta} \lf( 1 + \frac{2\eps k A(\eta)}{1-\eps k \eta} \ri) \fk^2(\eta)	\\
			 \geq - A(t_0) \fk^2(t_0) - C \fk^2( t_0) (\teps-  t_0) \geq - C |\log\eps| \fk^2( t_0).
		}
		Putting together the above bounds with \eqref{cost 6}, we obtain \eqref{cost 5} and thus the final result.
	\end{proof}

\section{Energy Estimate for General Domains}\label{sec:proof generic}

In this section we prove our main result for general domains, Theorem \ref{theo:generic}. We start in Subsection \ref{sec:restriction annulus} by a standard restriction to a thin boundary layer that we will also use in the disc case. We will be a little bit sketchy since only the refinement we will need for the disc case presents some novelty. 

Actually, the upper bound corresponding to \eqref{eq:main energy generic} has already been proven in \cite{AH}. The reader may easily check that nothing in the arguments therein requires any restriction on $\hex$ beyond \eqref{eq:b condition}, as far as the energy upper bound in concerned.
For shortness we state the result now, without proof:
\begin{equation}\label{eq:up bound generic}
\gle \leq \frac{|\dd \Om| \: \eoneo}{\eps} + \OO(1). 
\end{equation}
The method of proof, described in details in \cite[Section 14.4.2]{FH}, will anyway be employed in Section \ref{sec:proof disc} to obtain a refined bound in the case of the disc. Our main contribution here is to provide, in Subsection \ref{sec:generic low bound}, a lower bound matching \eqref{eq:up bound generic} in the whole regime \eqref{eq:b condition}, which will be given by the combination of \eqref{eq:low bound annulus} with \eqref{eq:low main generic}. The cornerstone of our new method is an energy splitting combined with the use of the cost function studied in Subsection \ref{sec:model half plane}.
 
\subsection{Restriction to the boundary layer and replacement of the vector potential}\label{sec:restriction annulus}

The first step in the proof of the lower bound is a restriction of the domain to a thin layer at the boundary of the sample together with a replacement of the minimizing vector potential $ \aavm $ with some explicit vector potential. Such a replacement (described in full details in \cite[Sections 14.4.1 \& Appendix F]{FH1}) can be made by exploiting some known elliptic estimates on solutions to the GL equations. 

We first introduce appropriate boundary coordinates: for any smooth simply connected domain $ \Omega $, we  denote by $ \gav(\xi): \R \setminus (|\partial \Omega| \Z) \to \partial \Omega $ a counterclockwise parametrization of the boundary $ \partial \Omega $ such that $ |\gav^{\prime}(\xi)| = 1 $. The unit vector directed along the inward normal to the boundary at a point $ \gav(\xi) $ will be denoted by $ \nuv(\xi) $. The curvature $ \tilde{k}(\xi) $ is then defined through the identity 
$$ \gav^{\prime\prime}(\xi) = \tilde{k}(\xi) \nuv(\xi). $$
Then we introduce the boundary layer where the whole analysis will be performed:
\beq
	\label{ann}
	 \annt : = \lf\{ \rv \in \Omega \: | \: \dist(\rv, \partial \Omega) \leq c_0 \eps |\log\eps| \ri\},
\eeq
and in such a region we can also introduce tubular coordinates $ (\xi,\eta) $ such that, for any given $ \rv \in \ann $, $ \eta = \dist(\rv, \partial \Omega) $, i.e.,
\beq
	\label{eq:tubular coordinates}
	\rv(\xi,\eta) = \gav^{\prime}(\xi) +  \eta \nuv(\xi),
\eeq
which can obviously be realized as a diffeomorphism for $\eps$ small enough. Hence the boundary layer becomes in the new coordinates $ (\xi,\eta) $ 
\beq
	\annpre := \left\{ (\xi,\eta) \: | \: \xi\in [0, |\partial \Omega| ],\;  \tau \in [0,c_0\eps|\log\eps|] \right\}.
\eeq

The energy in the boundary layer is given by the reduced GL functional 
\bml{\label{eq:GL BL pre}
	\annfpre[\psi] : = \int_{\ann} \diff s \diff t \: \lf(1 - \eps \curv t \ri) \lf\{ \lf| \partial_t \psi \ri|^2 + \frac{1}{(1- \eps \curv t)^2} \lf| \lf( \partial_s + i \aae(s,t) \ri) \psi \ri|^2 \ri.	\\	
	\lf. - \frac{1}{2 \hex} \lf[ 2|\psi|^2 - |\psi|^4 \ri]  \ri\},
}
where we have rescaled coordinates $ \xi = : \eps s $, $ \tau = : \eps t $ and set $ \curv : = \tilde{k}(\eps s) $. Moreover we define
\begin{equation}\label{eq:def ann rescale}
\ann:= \left\{ (s,t) \in \left[0, \tx\frac{|\partial \Omega|}{\eps} \right] \times \left[0,c_0 |\log\eps|\right] \right\},
\end{equation}
\beq
	\label{eq:aae}
	\aae(s,t) : = - t + \half \eps \curv t^2 + \eps \deps, 	
\eeq
and
\beq
	\deps : = \frac{\gamma_0}{\eps^2} - \lf\lfloor \frac{\gamma_0}{\eps^2} \ri\rfloor,	\qquad	 \gamma_0 : = \frac{1}{|\partial \Omega|} \int_{\Omega} \diff \rv \: \curl \aavm,
\eeq
$ \lf\lfloor \: \cdot \: \ri\rfloor $ standing for the integer part. Note that the three sets $ \annt, \annpre $ and $ \ann $ describe in fact the same domain but in different coordinates.

As discussed above we manage to use all the information contained in this functional only in the case where the domain $\Om$ is a disc, i.e., the curvature $k(s)$ is constant. For general domains we make a further simplification by dropping all the terms involving $k(s)$, which leads to the introduction of the functional
\beq\label{eq:GLf rescale generic}
	\annfO[\psi] : = \int_{\ann} \diff s \diff t \lf\{ \lf| \partial_t \psi \ri|^2 + \lf| \lf( \partial_s + i \aaO(s,t) \ri) \psi \ri|^2 - \frac{1}{2 \hex} \lf[ 2|\psi|^2 - |\psi|^4 \ri]  \ri\},
\eeq
where 
\beq\label{eq:def a0}
	\aaO(s,t) : =- t + \eps \deps.
\eeq
Note the similarity with the half-plane functional of Subsection \ref{sec:sketch}.

In the disc case we need to keep the (constant) curvature term and therefore we consider the functional
\begin{equation}\label{eq:GLf rescale disc}
 \annf[\psi] : = \int_{\ann} \diff s \diff t \lf(1 - \eps k t \ri) \lf\{ \lf| \partial_t \psi \ri|^2 + \frac{1}{(1- \eps k t)^2} \lf| \lf( \partial_s + i \aae(s,t) \ri) \psi \ri|^2 \ri.	\\	
	\lf. - \frac{1}{2 \hex} \lf[ 2|\psi|^2 - |\psi|^4 \ri]  \ri\},
\end{equation}
which differs from $ \annfpre $ given in \eqref{eq:GL BL pre} only because of the constant curvature $ k(s) \equiv k $, so that with a little abuse of notation we keep denoting
\beq\label{eq:def aeps}
	\aae(s,t) : =- t + \half \eps k t^2 + \eps \deps. 
\eeq
To obtain these simplified functionals from the full GL energy requires to extract the $\phi_\eps$ phase factor from the GL minimizer. We now recall its definition:
\beq
			\label{phieps}
			\phi_{\eps}(s,t) : = - \frac{1}{\eps} \int_{0}^{t} \diff \eta \: \nuv(\eps s) \cdot \aavm(\rv(\eps s, \eps \eta)) + \frac{1}{\eps} \int_{0}^s \diff \xi \: \gav^{\prime}(\eps \xi) \cdot \aavm(\rv(\eps \xi,0)) - \eps \deps s.			\eeq	
The link between the functionals in boundary coordinates and the original GL functional is given by the following 

	\begin{pro}[\textbf{Replacement of the vector potential and restriction to $\ann$}]
		\label{replacement: pro}
		\mbox{}	\\
		For any smooth simply connected domain $\Om$ and $1<b<\theo ^{-1}$, in the limit $\eps\to 0$, one has the lower bound
		\beq
			\label{eq:low bound annulus}
			\glee \geq  \annfO[\psi] - C,
		\eeq
		where $ \psi(s,t) \in H^1(\ann) $ 
		vanishes at the inner boundary of $ \ann $,
		\beq
			\label{eq:boundary conditions rescale}
			\psi(s,c_0|\log\eps|) = 0,	\quad \mbox{for any } s \in \lf[0, \tx\frac{|\partial \Omega|}{\eps} \ri].
		\eeq
		In the disc (constant curvature) case the energy lower bound reads
		\beq
			\label{eq:low bound annulus disc}
			\glee \geq \annf[\psi] -C \eps^2 |\log\eps|^2,
		\eeq
		with 
		$$ \psi(s,t) = \glm(\rv(\eps s,\eps t))  e^{-i\phi_{\eps}(s,t)} \mbox{ in } \ann .$$
	\end{pro}


	\begin{proof}
		The first part of the result, i.e., the restriction to the boundary layer is in fact already proven in details in \cite[Section 14.4]{FH1}: the combination of a suitable partition of unity together with standard Agmon estimates on the decay of $ \Psi $ far from the boundary yield
		\beq
			\glee \geq \int_{\annt} \diff \rv \bigg\{ \bigg| \bigg( \nabla + i \frac{\aavm}{\eps^2} \bigg) \Psi_1 \bigg|^2 - \frac{1}{2 \hex \eps^2} \lf[ 2|\Psi_1|^2 - |\Psi_1|^4 \ri]  \bigg\} + \OO( \eps^{\infty}).
		\eeq
		Here $ \Psi_1 $ is given in terms of $ \glm $ in the form $ \Psi_1 = f_1 \glm $ for some radial $ 0 \leq f_1 \leq 1 $ with support containing the set $ \annt $ defined by \eqref{ann} and contained in  
		\bdm
			\{ \rv \in \Omega \: | \: \dist(\rv, \partial \Omega) \leq C \eps|\log\eps| \}
		\edm
		 for a possibly large constant $ C $. Note that such an estimate requires that the constant $ c_0 $ occurring in the definition \eqref{ann} of the boundary layer be chosen large enough. However the choice of the support of $ f_1 $ remains to any other extent arbitrary and one can clearly pick $ f_1 $ in such a way that it vanishes at the inner boundary of $ \annt $, implying \eqref{eq:boundary conditions rescale}. On the other hand, in the disc case, we can stick to the choice made in \cite[Section 14.4]{FH1}, e.g., $ f_1 = 1 $ in $ \annt $ and going smoothly to $ 0 $ outside of it.

		The second step is then the replacement of the magnetic vector potential $ \aavm $ and a simultaneous change of coordinates $ \rv \to (s,t) $.  We use the change of gauge induced by
		$$ \Psi_1(\rv) = \psi(s,t) e^{ i \phi_{\eps}(s,t)}. $$
		The function $ \psi $ is clearly single-valued since, for any given $ t $, $ \phi_\eps (s + n \frac{|\partial \Omega|}{\eps}, t) = \phi_\eps (s,t) + 2 \pi n_{\eps} $ for some integer number $ n_{\eps} \in \Z $. Moreover, by gauge invariance,
		\bml{
			\label{magnetic identity}
			\int_{\annt} \diff \rv \bigg| \bigg( \nabla + i \frac{\aavm}{\eps^2} \bigg) \Psi_1 \bigg|^2 	= \int_{\ann} \diff s \diff t \: (1 - \eps \curv t) \lf\{  \big| \partial_t  \psi \big|^2 \ri. \\
			\lf. + \frac{1}{(1 - \eps \curv t)^2} \lf| \lf(  \partial_s + i \ate(s,t) \ri) \psi \ri|^2  \ri\},
		}
		with
		\bml{
			\ate(s,t) = (1 - \eps \curv t) \frac{\gav^{\prime}(\eps s) \aavm(\rv(\eps s,\eps t))}{\eps} - \partial_s \phi_{\eps}  
			=  (1 - \eps \curv t) \frac{\gav^{\prime}(\eps s) \aavm(\rv(\eps s,\eps t))}{\eps}	\\
			+ \frac{1}{\eps} \int_{0}^{t} \diff \eta \: \partial_s \lf[ \nuv(\eps s) \cdot \aavm(\rv(\eps s, \eps \eta)) \ri]	   - \frac{\gav^{\prime}(\eps s) \cdot \aavm(\rv(\eps s,0))}{\eps} + \eps \deps,
		}
		since the normal component of the new vector potential vanishes:
		\bdm		
			\frac{\nuv(\eps s) \cdot\aavm(\rv(\eps s, \eps t))}{\eps} - \partial_t \phi_{\eps} = 0.
		\edm
		Hence we also have
		\beqn
			\ate(s,0) & = &  \eps \deps,	\nonumber	\\
			\partial_t \ate(s,t) & = & - (1 - \eps \curv t) \: \lf( \curl \aavm \ri) (\rv(\eps s,\eps t)) = - (1 - \eps \curv t) (1 + f(s,t)).
		\eeqn
		Here we have set
		\beq
			\lf( \curl \aavm \ri) (\rv(\eps s,\eps t)) = : 1 + f(s,t),
		\eeq
		for some smooth function $ f $ such that $ f(s,0) = 0 $, thanks to the boundary condition \eqref{eq:GL bc}. Therefore
		\beq
			\label{tildeas}
			\ate(s,t) =  \eps \deps - \int_0^t \diff \eta \: (1 - \eps \curv \eta) \lf(1 +  f(s,\eta) \ri).
		\eeq
		Hence in order to complete the replacement, we only have to replace $ \ate $ with $ \aae $ given by \eqref{eq:aae}, which amounts to estimating the last term in \eqref{tildeas}. We start from the simple inequality
		\bdm
 			\lf\| \ate - \aae \ri\|_{L^{\infty}(\ann)} \leq C |\log\eps| \lf\| f \ri\|_{L^{\infty}(\ann)},
		\edm
		but as shown in \cite[Proof of Lemma F.1.1]{FH1} one can also prove that 
		\bdm
			 \lf|f(s,t) \ri| \leq C \eps t \lf\| \nabla \curl \aavm \ri\|_{C^0(\annt)}.
		\edm
		This in turn yields $  \lf\| f \ri\|_{L^{\infty}(\ann)} \leq C \eps^2|\log\eps| $ and
		\beq
			\label{eq:diff magn fields}
			\lf\| \ate - \aae \ri\|_{L^{\infty}(\ann)} =  \OO(\eps^3 |\log\eps|^2),
		\eeq
		via the inequality \cite[Eq. (11.51)]{FH1}
		\beq
			\label{eq: elliptic est}
			\lf\| \curl \aavm - 1 \ri\|_{C^1(\overline{\Omega})} =  \OO(\eps).
		\eeq

	 	In the disc case ($ k(s) \equiv k $ constant), the bound \cite[Eq. (10.21)]{FH1}
		\beq
			\label{linfty psi}
			\lf\| \glm \ri\|_{\infty} \leq 1,	
		\eeq
		in combination with \eqref{eq:diff magn fields} is sufficient to obtain the result \eqref{eq:low bound annulus disc}:
		\bml{
 			\int_{\ann} \diff s \diff t \: \frac{1}{1 - \eps k t} \lf| \lf(  \partial_s + i \ate \ri) \psi \ri|^2 - \int_{\ann} \diff s \diff t \: \frac{1}{1 - \eps k t} \lf| \lf(  \partial_s + i \aae \ri) \psi \ri|^2  \\
			=- 2 \Im \int_{\ann} \diff s \diff t \: \frac{1}{1 - \eps k t} \lf[ \lf( \partial_s + i \ate \ri) \psi \ri]^*  \lf( \ate - \aae \ri) \psi - \int_{\ann} \diff s \diff t \: \frac{1}{1 - \eps k t} \lf| \ate - \aae \ri|^2 \big| \psi \big|^2	\\
		\geq - \delta \int_{\ann} \diff s \diff t \: \frac{1}{1 - \eps k t} \lf| \lf(  \partial_s + i \ate \ri) \psi \ri|^2  - \lf(\frac{1}{\delta} + 1\ri) \int_{\ann} \diff s \diff t \: \frac{1}{1 - \eps k t} \lf| \ate - \aae \ri|^2 \big| \psi \big|^2 	\\
		\geq - \delta \int_{\annt} \diff \rv \bigg\{ \bigg| \bigg( \nabla + i \frac{\aavm}{\eps^2} \bigg) \Psi_1 \bigg|^2 - C \lf(\frac{1}{\delta} + 1\ri) \eps^6 |\log\eps|^4 \int_{\ann} \diff s \diff t \: (1 - \eps k t) \big| \psi|^2 \\
		\geq  - C \bigg[ \frac{\delta}{\eps} + \frac{\eps^4 |\log\eps|^4}{\delta} \int_{\annt} \diff \rv \: \lf| \glm \ri|^2  \bigg]  \geq - C \lf[ \frac{\delta}{\eps} +  \frac{C \eps^5 |\log\eps|^4}{\delta} \ri] \geq -  C\eps^2 |\log\eps|^2.
 		}
		In the estimates above we reconstructed the GL kinetic energy by means of \eqref{magnetic identity}, used the inequality 
		(Agmon estimate \cite[Eq. (12.9)]{FH}, $A$ is a fixed constant)
		\bml{
			\label{est grad psi}
			\int_{\Om} \diff \rv \: \exp\lf\{ \tx\frac{A \: \dist(\rv, \partial \Omega)}{\eps} \ri\} \bigg\{ \lf| \glm \ri|^2 + \eps^2 \bigg| \bigg( \nabla + i \frac{\aavm}{\eps^2} \bigg) \glm \bigg|^2 \bigg\}	\\
			\leq \int_{ \dist(\rv, \partial \Omega) \leq \eps} \diff \rv \: \lf| \glm \ri|^2 =  \OO(\eps),
 		}
		and finally optimized over $ \delta $. Note that here we have used the assumption $ b > 1 $ to apply \cite[Theorem 12.2.1]{FH}; \eqref{est grad psi} is wrong without this assumption.
	
		For general domains a rougher estimate is even sufficient:
 		\bml{
 			\int_{\ann} \diff s \diff t \: \frac{1}{1 - \eps k t} \lf| \lf(  \partial_s + i \ate \ri) \psi \ri|^2 - \int_{\ann} \diff s \diff t \: \frac{1}{1 - \eps k t} \lf| \lf(  \partial_s + i \aae \ri) \psi \ri|^2	\\
			\geq - C \eps^3 |\log\eps|^2 \int_{\ann} \diff s \diff t \: \lf[ \lf| \partial_s \psi \ri| + \lf| \aae \ri| \lf| \psi \ri| \ri] \geq  - C \eps^{5/2}|\log\eps|^{3/2} \lf[ \lf\| \partial_s \psi \ri\|_{L^2(\ann)} + \lf\| t \psi \ri\|_{L^2(\ann)} \ri]	\\
			 \geq -  C\eps^2 |\log\eps|^{3/2},
		}
		thanks again to the Agmon estimate \eqref{est grad psi}, which after rescaling becomes
		\beq
			\label{eq:est grad psi rescaled}
			\int_{\ann} \diff s \diff t \: (1 - \eps \curv t) \: e^{ A t } \lf\{ \lf| \psi(s,t) \ri|^2 
			+ \lf| \lf( \nabla_{s,t} + i \ate(s,t) {\bf e}_s \ri) \psi \ri|^2 \ri\} =  \OO(\eps^{-1}).
 		\eeq
		Note that, since $ \psi(s,t) = \lf(f_1 \glm\ri)(\rv(\eps s, \eps t)) e^{-i\phi_{\eps}(s,t)} $, one also needs to control the terms involving the gradient of $ f_1 $. However thanks to the freedom in the choice of the support of $ f_1 $ as well as its smoothness, we can always assume that such terms are smaller or at most of the same order as the error appearing on the r.h.s. of the expression above.

		To complete the proof for general domains, it remains to estimate the errors due to the curvature terms dropped in \eqref{eq:GLf rescale generic}, but one can easily realize that such terms are all bounded by quantities of the form
		\beq
			\eps \lf| \int_{\ann} \diff s \diff t \: \curv t \lf\{ \lf| \nabla \psi \ri|^2 + \lf| \psi \ri|^2 \ri\} \ri| \leq C \eps \int_{\ann} \diff s \diff t \: t \lf\{ \lf| \nabla \psi \ri|^2 + \lf| \psi \ri|^2 \ri\} =  \OO(1),
		\eeq
		again by \eqref{eq:def aeps}, \eqref{eq:diff magn fields} and \eqref{eq:est grad psi rescaled}.		
	\end{proof}

\subsection{Lower bound in the boundary layer}\label{sec:generic low bound}

We now provide our main new argument in the case of general domains, namely we bound the rescaled functional $\annfO$ from below:

\begin{pro}[\textbf{Lower bound to $\annfO $}]\label{pro:low bound generic}\mbox{}\\
For any $1<b<\theo ^{-1}$, $\eps$ small enough and $\psi \in H^1(\ann) $  satisfying the boundary condition \eqref{eq:boundary conditions rescale}, it holds 
\begin{equation}\label{eq:low main generic}
\annfO [\psi] \geq \frac{|\dd \Om| \: \eoneo}{\eps}.
\end{equation}
\end{pro}


A crucial ingredient of the proof of the above result is the following lemma. We are going to use the notation
\begin{equation}\label{eq:current u}
(iu,\dd_s u) := \tx\frac{i}{2} \left( u \dd_s u ^* - u ^* \dd_s u\right)   
\end{equation}
for the $s$-component of the {\it superconducting current} associated with $u$.

\begin{lem}[\textbf{Energy splitting for general domains}]\label{lem:split general}\mbox{}\\
Under the assumptions of Proposition \ref{pro:low bound generic}, define $u(s,t)$ by setting  
\beq\label{eq:split func generic}
\psi(s,t) =  u(s,t)  \fO (t) \exp \lf\{ - i \lf( \alO + \eps \deps\ri) s \ri\}.
\eeq
Then one has
\begin{equation}\label{eq:split ener generic}
\annfO[\psi] = \frac{|\dd \Om| \: \eoneo}{\eps} + \E_0 [u],
\end{equation}
where
\begin{equation}
\label{eq:energy E0}
\E_0 [u] : =  \int_{\ann} \diff s \diff t \: \fO ^2(t) \: \bigg\{ \lf| \partial_t u \ri|^2 + \lf| \partial_s u \ri|^2 - 2 (t+\alO) (iu,\dd_s u) + \frac{1}{2 \hex} \fO^2(t) \lf( 1 - |u|^2 \ri)^2 \bigg\}.
\end{equation}
\end{lem}

\begin{proof}
Note first that since $1<\hex<\theo ^{-1}$, $\fO$ is strictly positive everywhere, so that \eqref{eq:split func generic} makes sense. Because of the phase factor, $u$ needs not be periodic in the $s$ variable, contrarily to $\psi$, but this will be of no concern to us since it is sufficient for the rest of the argument that $|u|$ is periodic. We have 
\[
\int_{\ann} \diff s \diff t \: |\dd_t \psi| ^2 = \int_{\ann} \diff s \diff t \: \lf\{ \fO ^2 |\dd_t u| ^2 + \fO \dd_t \fO  \dd_t |u|^2  + |u| ^2 |\dd_t \fO| ^2 \ri\}, \]
and an integration by parts in $t$ shows that 
\[
\int_{\ann} \diff s \diff t \: \fO \dd_t \fO \dd_t |u|^2 = - \int_{\ann} \diff s \diff t \: |u| ^2 |\dd_t \fO| ^2 - \int_{\ann} \diff s \diff t \: \fO |u| ^2  \dd_t ^2 \fO.
\]
Indeed  the boundary terms vanish because of the Neumann boundary condition satisfied by $\fO$ at $t=0$ and the Dirichlet condition satisfied by $u$ at $t= c_0 |\log \eps|$, inherited from \eqref{eq:boundary conditions rescale}. Then inserting the variational equation \eqref{eq:1D half plane vareq}, we obtain
\[
\int_{\ann} \diff s \diff t \: |\dd_t \psi| ^2 = \int_{\ann} \diff s \diff t \: \fO ^2 |\dd_t u| ^2 + \int_{\ann} \diff s \diff t \: \fO |u| ^2 \left( - \potO \fO + \frac{1}{\hex} \fO \left( 1 - \fO ^2 \right)\right).
\]
On the other hand, by the definition of $\potO(t) = (\alO + t) ^2 $ and~\eqref{eq:def a0},
\[
\int_{\ann} \diff s \diff t \: \left|\left( \dd_s + i \aaO \right) \psi\right| ^2 =  \int_{\ann} \diff s \diff t \lf\{ \fO ^2 |\dd_s u | ^2 + \potO \fO ^2 |u| ^2 - 2 (t+\alO) (iu,\dd_s u) \ri\}.
\]
Hence, combining all the above equalities with the $k=0$, $\eps =0$ version of \eqref{eone explicit} and recalling that we work on a rectangle whose length in the $s$ direction is $ |\dd \Om|/\eps$, we obtain the desired formula.
\end{proof}

We can now conclude the 

\begin{proof}[Proof of Proposition \ref{pro:low bound generic}]
In this proof $\psi$ denotes the GL order parameter after the gauge choice, localization to the boundary and change of coordinates we have described up to now. We define $u$ as in \eqref{eq:split func generic} and start from \eqref{eq:split ener generic}, so that we only have to bound $\E_0 [u]$ from below.

We integrate by parts the momentum term (second term of \eqref{energy E}), using the potential function $ \FO(t) $ defined in \eqref{Fk}, i.e.,
$$ \FO (t) = 2 \int_0^t \diff \eta \: \left( \eta + \alO \right) \fO^2(\eta)$$,
which satisfies $\FO(0) = 0$ and 
$$ \FO ' (t) =  2  \left( t + \alO \right) \fO^2( t).$$ 
This gives
\begin{multline} 
\label{eq:disc l b step 1}
 - 2   \int_{\ann} \diff s \diff t \:  \fO^2(t) \left( t + \alO \right) \lf( i u, \partial_s u \ri)= - \int_{\ann} \diff s \diff t \:  \partial_t \FO (t) \lf( i u, \partial_s u \ri)\\
=  \int_{\ann} \diff s \diff t \:  \FO (t) \partial_t \lf( i u, \partial_s u \ri)
\end{multline}
by integrating by parts in the $t$ variable. Boundary terms vanish because $\FO = 0$ and $u=0$ respectively at $t=0$ and $t=c_0 |\log \eps|$.

A further integration by parts in $s$ then yields that for each fixed $t$
\begin{multline*}
\int_{0}^{\frac{2\pi}{\eps}} \diff s \: \partial_t \lf( i u, \partial_s u \ri) = \int_{0}^{\frac{2\pi}{\eps}} \diff s \: \lf[ i \partial_t u \partial_s u^* - i \partial_t u^* \partial_s u \ri] + \left[ u \dd_t u ^* + u ^* \dd_t u \right]_0 ^{\frac{|\dd \Om|}{\eps} } 
\\ = \int_{0}^{\frac{2\pi}{\eps}} \diff s \: \lf[ i \partial_t u \partial_s u^* - i \partial_t u^* \partial_s u \ri].
\end{multline*}
The boundary terms vanish because, since $\psi$ is periodic in the $s$ variable (by continuity of the GL order parameter), so is $u \dd_t u ^*$, as one can easily check. We may then write 
\bml{
 \label{eq:l b step 0}
- 2   \int_{\ann} \diff s \diff t \:  \fO^2(t) \left( t + \alO \right) \lf( i u, \partial_s u \ri) = \int_{\ann} \diff s \diff t \:  \FO (t) \partial_t \lf( i u, \partial_s u \ri) \geq - 2 \int_{\ann} \diff s \diff t \: \lf| \FO (t) \ri| \lf|  \partial_t u \ri| \lf| \partial_s u \ri|	\\
\geq  \int_{\ann} \diff s \diff t \: \FO (t) \lf[ \lf| \partial_t u \ri|^2 + \lf| \partial_s u \ri|^2 \ri],
}
where we have used the inequality $ab\leq (a^2 + b ^2) /2$ and the fact that $ \FO (t) $ is negative for any $ (s,t) \in \ann $ (see Lemma~\ref{lem:F prop}). 

Inserting the above lower bound in \eqref{eq:energy E0}, we get
\begin{equation}
 \label{eq:l b step 2}
\E_0[u] \geq  \int_{\ann} \diff s \diff t \: \KO(t) \lf[ \lf| \partial_t u \ri|^2 +\lf| \partial_s u \ri|^2 \ri] + \frac{1}{2 \hex}  \int_{\ann}\diff s \diff t \fO^4 \lf( 1 - |u|^2 \ri)^2,
\end{equation}
where $ \KO (t) $ is the cost function defined in \eqref{K0}. By Proposition \ref{pro:K positive}, $ \KO $ is positive in $ \ann $ and we can thus drop the first term as far as a lower bound is concerned, obtaining
\begin{equation}\label{eq:l b final}
 \E_0 [u] \geq \frac{1}{2 \hex}  \int_{\ann}\diff s \diff t \fO^4 \lf( 1 - |u|^2 \ri)^2 \geq 0,
\end{equation}
which leads to \eqref{eq:low main generic}.
\end{proof}

Note that putting \eqref{eq:up bound generic}, \eqref{eq:low bound annulus}, \eqref{eq:split func generic}, \eqref{eq:split ener generic} and \eqref{eq:l b final}  together, we obtain as a by-product 
\begin{equation}\label{eq:density generic pre}
\frac{1}{2 \hex}  \int_{\ann}\diff s \diff t \left( \fO ^2 - |\psi| ^2 \right) ^2 = \frac{1}{2 \hex}  \int_{\ann}\diff s \diff t \fO^4 \lf( 1 - |u|^2 \ri)^2 =  \OO(1), 
\end{equation}
which is essentially \eqref{eq:main density generic}. Indeed, once expressed in the original variables, this yields the estimate 
\[
\left\Vert |\glm| ^2 - \left| \fO \left( \tx\frac{\eta}{\eps} \right) \right| ^2 \right\Vert_{L ^2 (\annt)} =  \OO(\eps^2), 
\]
but using the usual decay estimates, one easily sees that the region $\Om \setminus \annt$ contributes a $\OO(\eps ^{\infty})$ to the integral, and also that (recall \eqref{fal point l u b}) $$\left\Vert \left| \fO \left( \tx\frac{\eta}{\eps} \right) \right| ^2 \right\Vert_{L ^2 (\annt)} \geq C \eps ^{1/2}. $$

\section{Surface Behavior for Disc Samples}\label{sec:proof disc}

The refinements we obtain in the disc case are based primarily on the fact that, the curvature being constant, the effective problem obtained from the ansatz \eqref{eq:ansatz} stays 1D even when considering subleading corrections to the energy. This allows us to go one term further in the energy expansion, which is the content of Theorem \ref{theo:disc energy}, whose proof is provided in Subsections \ref{sec:disc ub} and \ref{sec:disc lb}. As a by-product we obtain a refined control of crucial terms in the reduced energy of any minimizer, that lead to the proof of Theorem \ref{theo:Pan}, once combined with appropriate gradient estimates. Indeed, estimates for general domains in \eqref{eq:main density generic} are still compatible with small normal inclusions in the boundary layer. The refinements we obtain in the case of the disc rule out such inclusions once combined with some natural bounds on the gradient of the order parameter, proven in \cite{Al}.

Theorem \ref{theo:disc energy} is proven by comparing appropriate upper and lower bounds to the GL energy. The upper bound is discussed in Proposition \ref{pro:up bound disc}, whereas the lower bound is obtained by combining \eqref{eq:low bound annulus disc} in Proposition \ref{replacement: pro} with the result of Proposition \ref{pro:low bound disc} below.

\subsection{Energy upper bound}\label{sec:disc ub}




The upper bound part of the proof of Theorem \ref{theo:disc energy} is as usual the easiest one to get, since it suffices to provide a suitable trial configuration $ (\trial,\atrial) \in \gldom$ to test the GL energy. We already know that, according to the magnetic field replacement and restriction to the boundary layer discussed in Subsection \ref{sec:restriction annulus}, the vector potential $ \atrial $ should be equal or close enough to the one associated with the applied field. Additionally the order parameter $ \trial $ should be essentially supported in the boundary layer $ \annt $. Moreover, as suggested by our heuristic analysis, the modulus of $ \trial $ should be given by the 1D profile minimizing \eqref{eq:1D func disc} with $ \alpha $ equal to the optimal phase $ \alk $. Finally the winding number of $ \trial $ should be approximated by the optimal value $ \alk $ (actually $ \frac{\alk}{\eps} $ after the proper rescaling), so that the phase of $ \trial $ should have the form $ i\frac{
\alk}{\eps} \vartheta $. 
Now this last requirement is the only non-trivial one to meet, since the trial function $ \trial $ must be a single-valued function. Unfortunately in general $ \frac{\alk}{\eps} $ is not an integer and therefore not an allowed phase. To overcome this problem we prove the following

	\begin{lem}[{\bf Optimal integer phase}]
		\label{lem:opt int phase}
		\mbox{}	\\
		For any $ k > 0  $ and $ \eps$ small enough, let $ \alk \in \R $ and $ \eonek $ be defined as in \eqref{eq:optimal energy} and set
		\beq
			\alkt : = \eps \lf\lfloor \frac{\alk}{\eps} \ri\rfloor.
		\eeq
		Then one has
		\beq
			\label{eq:opt int phase}
			\eonek \leq \fone_{k,\alkt}[\fk]  = \eonek + \OO(\eps^2|\log\eps|).
		\eeq
	\end{lem}
	
	\begin{proof}
		The inequality $ \eonek \leq  \eone_{k,\alkt} $ is trivial. Exploiting the optimality of $ \alk $ \eqref{FH nonlinear} as well as the bound \eqref{fal estimate}, one obtains the opposite inequality (recall that $f_k = f_{k,\alk}$):
		\bml{
 			\fone_{k,\alkt}[\fk] = \eonek + \int_{0}^{\teps} \diff t \: (1 - \eps k t)^{-1} \lf\{ \alkt^2 - \alk^2 + 2 \lf(\alkt - \alk\ri) \lf(t - \tx\frac{1}{2} \eps k t^2 \ri) \ri\} \fk^2	\\
			= \eonek +  \lf( \alkt - \alk \ri)^2 \int_{0}^{\teps} \diff t \: (1 - \eps k t)^{-1} \fk^2 = \OO(\eps^2|\log\eps|),
		}
		since $ (\alkt - \alk)^2 = \OO(\eps^2) $.
	\end{proof}

	\begin{pro}[\textbf{Energy upper bound for disc samples}]
		\label{pro:up bound disc}
		\mbox{}	\\
		Let $\Om$ be a disc of radius $R=k ^{-1} > 0 $. For any $ 1 < \hex < \theo^{-1} $ and $ \eps$ sufficiently small, we have the upper bound
		\beq
			\label{eq:up bound disc}
			\glee \leq \frac{2\pi R \: \eonek}{\eps} + C\eps|\log\eps|.
		\eeq
	\end{pro}

	\begin{proof}
		We denote radial coordinates by $(r,\vartheta)$ and define our trial order parameter as
		\beq
			\label{trial}
			\trial(\rv) : = 
			\begin{cases}
				\fk\lf(\tx\frac{1-r}{\eps}\ri) \exp\lf\{-i \lf\lfloor \tx\frac{\alk}{\eps} \ri\rfloor \vartheta \ri\},		&	\mbox{in } \annt,	\\
				\fk(\teps) \chi(r) \exp\lf\{-i \lf\lfloor \tx\frac{\alk}{\eps} \ri\rfloor \vartheta \ri\},				&	\mbox{otherwise},
			\end{cases}
		\eeq
		with $ \chi $ a smooth cut-off function such that $ \chi(1-\eps \teps) = 1 $ and going to zero exponentially fast for smaller $ r $. The trial magnetic field on the other hand has to be close to the minimizing one and we then pick
		\beq
			\atrial(\rv) : =  - \frac{R^2 - r^2}{2 r} \mathbf{e}_{\vartheta}
		\eeq
		as our trial vector potential. 
		Now, thanks to the exponential smallness of $ \fk(\teps) $ in $ \eps $, we can always choose the cut-off $ \chi $ in such a way that the contribution to the GL energy coming from outside $ \annt $ is itself $ \OO(\eps^{\infty}) $. Hence a rather straightforward computation yields
		\beq
			\glfe \lf[ \trial, \atrial \ri] = \annft \lf[\fk(t) \exp\lf\{-i \alkt s \ri\} \ri] + \OO(\eps^{\infty}),
		\eeq
		with $ \annft $ standing for the functional \eqref{eq:GLf rescale disc}, where $ \aae $ has been replaced with
		\beq
			- t + \half \eps t^2.
		\eeq
		Then it is trivial to verify that
		\beq
			\annft \lf[\fk(t) \exp\lf\{-i \alkt s \ri\} \ri] = \frac{2 \pi R \: \fone_{k,\alkt}[\fk] }{\eps}\leq  \frac{2 \pi R \: \eonek}{\eps} + C\eps|\log\eps|,
		\eeq
		thanks to Lemma \ref{lem:opt int phase} above.		
	\end{proof}

\subsection{Energy lower bound}\label{sec:disc lb}

Our starting point is the inequality \eqref{eq:low bound annulus disc}, which allows to restrict the GL energy to the (rescaled) boundary layer $ \ann $. In this section we present a refined analysis of the functional $\annf$, which leads to the following main result:

\begin{pro}[\textbf{Lower bound to $\annf $}]\label{pro:low bound disc}\mbox{}\\
Let $\Om$ be a disc of radius $R=k ^{-1} > 0 $. For any $ 1 < \hex < \theo^{-1} $, $c_0$ large enough and 
$$\psi(s,t) = \glm(\rv(\eps s,\eps t)) e^{-i\phi_{\eps}(s,t)}, $$
we have, in the limit $\eps \to 0$, 
\begin{equation}\label{eq:low bound disc}
\annf [\psi] \geq \frac{2 \pi R \: \eonek}{\eps} + \OO(\eps^{\infty}).
\end{equation}
\end{pro}


As for general domains the first step towards the proof of the above result is an improved energy splitting associating $\annf$ with a reduced energy functional where $\fk$ appears as a weight. Recall the definition of the superconducting current given in \eqref{eq:current u}.

	\begin{lem}[\textbf{Energy splitting for disc samples}]
		\label{lem:split disc}
		\mbox{}	\\
		Under the assumptions of Proposition \ref{pro:low bound disc}, define $u(s,t)$ by setting 
		\beq
			\label{def u}
			\psi(s,t) =  \fk ( t) \: u( s,t) \: \exp \lf\{ - i \left(\alk+  \eps \deps \ri) s\ri\},
		\eeq
		where $\deps$ is defined in \eqref{eq:deps}. Then one has the identity
		\beq
			\label{splitting}
			\annf[\psi] = \frac{2\pi R \: \eonek}{\eps} + \E_k [u],
		\eeq
		where
		\bml{
			\label{energy E}
			\E_k [u] : =  \int_{\ann} \diff s \diff t \: \lf(1 - \eps k t \ri) \fk^2(t) \bigg\{ \lf| \partial_t u \ri|^2 + \frac{1}{(1- \eps k t)^2} \lf| \partial_s u \ri|^2 \\
			 - 2  \beps(t) \lf( i u, \partial_s u \ri) + \frac{1}{2 \hex} \fk ^2(t) \lf( 1 - |u|^2 \ri)^2 \bigg\},
		}
		\beq
			\label{beps}
			\beps(t) : = \frac{t + \alk - \half \eps k t^2}{(1 - \eps k t)^2}.
		\eeq
	\end{lem}

	\begin{proof}
	Note that \eqref{def u} makes sense in view of Corollary \ref{nontrivial: cor} and the assumption $1 < \hex < \theo^{-1} $. The proof is essentially identical to that of Lemma \ref{lem:split general}. Note that, as in Lemma \ref{lem:split general}, $ u $ needs not be periodic in the $ s $ variable, since in general $ \frac{\alk}{\eps} + \deps $ is not an integer number. However this does not matter in the proof of \eqref{splitting} and for later purposes (see the proof of Proposition \ref{pro:low bound disc}) we stress that, unlike $ u $ itself, $ u^* \partial_t u $ is periodic in $ s $ with period $ \frac{2\pi}{\eps} $. 

	The key point is again the computation of the $t$-component of the kinetic energy, which proceeds as follows:
	\bml{
	\int_{\ann} \diff s \diff t \: \lf(1 - \eps k t \ri) \lf| \partial_t \psi \ri|^2 = \int_{\ann} \diff s \diff t \: \lf(1 - \eps k t \ri) \left\{ \fk ^2 |\dd_t u| ^2 + \fk \dd_t \fk  \dd_t \lf| u \ri|^2  + |u| ^2 |\dd_t \fk| ^2\right\},\\
	= \int_{\ann}\diff s \diff t \: \lf(1 - \eps k t \ri) \fk ^2 |\dd_t u| ^2 - \int_{\ann} \diff s \diff t \: \fk |u| ^2 \dd_t \left( \lf(1 - \eps k t \ri) \dd_t \fk\right)\\
	= \int_{\ann}\diff s \diff t \: \lf(1 - \eps k t \ri) \fk ^2 |\dd_t u| ^2 + \int_{\ann}\diff s \diff t \: \lf(1 - \eps k t \ri) \fk ^2 |u| ^2  \left( - \potk \fk + \tx\frac{1}{\hex} \fk \left( 1 - \fk ^2 \right)\right),
	}
	where we have used the Neumann boundary conditions satisfied by $\fk$ at both boundaries of the interval $\ieps$ and the variational equation \eqref{var eq fal} for $\al = \alk$ rewritten as
	\[
	 - \dd_t \left( \lf(1 - \eps k t \ri) \dd_t \fk\right) + \left( 1-\eps kt \right) \fk \left( \potk - \tx\frac{1}{\hex} \fk \left( 1 - \fk ^2 \right)\right) = 0. 
	\]
	The computation of the other terms in the functional is trivial and there only remains to group them properly, use \eqref{eone explicit} and the definition \eqref{pot} to conclude.
	\end{proof}

We may now complete the proof of the lower bound \eqref{eq:low bound disc}. By Lemma \ref{lem:split disc}, it suffices to bound from below the energy $ \E_k[u] $ given by \eqref{energy E}. In comparison with the corresponding step in Subsection \ref{sec:generic low bound}, we face the technical difficulty that we do not know that the relevant cost function is positive in the whole domain but only in 
\begin{equation}\label{eq:annbb}
 \annbb := \left\{ (s,t) \in \ann, \: t \leq \btik \right\}
\end{equation}
where $\btik$ is defined at the beginning of Section \ref{sec:model disc} and satisfies \eqref{bxi}. The complementary region $\ann \setminus \annbb$ is dealt with by more ``brute force'' estimates, showing that its final contribution to the energy is at most $\OO (\eps ^{\infty})$.

\begin{proof}[Proof of Proposition \ref{pro:low bound disc}] 
In this proof $\psi$ and $u$ are defined starting from a GL minimizer as in the statement of Proposition \ref{pro:low bound disc}. We again start by rewriting the momentum term (second term of \eqref{energy E}). The potential function $ \Fk(t) $ defined in \eqref{Fk}  satisfies 
		$$ \partial_t \Fk (t) =  2  (1- \eps k t ) \beps(t) \fk^2( t),$$
		and $\Fk (0) = \Fk (\teps) = 0$, so that
		\begin{equation} 
			\label{l b step 1}
 			- 2   \int_{\ann} \diff s \diff t \: \lf(1 - \eps k t \ri) \fk^2(t) \beps(t) \lf( i u, \partial_s u \ri) = \int_{\ann} \diff s \diff t \:  \Fk (t) \partial_t \lf( i u, \partial_s u \ri).
		\end{equation}
		As before, a further integration by parts in the $s$ yields for each fixed $t$ (here we are using the $s$-periodicity of $ u^* \partial_t u $)
		\beq
			\int_{0}^{\frac{2\pi}{\eps}} \diff s \: \partial_t \lf( i u, \partial_s u \ri) = \int_{0}^{\frac{2\pi}{\eps}} \diff s \: \lf[ i \partial_t u \partial_s u^* - i \partial_t u^* \partial_s u \ri]. 
		\eeq
		We can thus easily estimate inside $ \annbb $
		\bml{
 			\label{l b step 0}
			 \int_{\annbb} \diff s \diff t \:  \Fk (t) \partial_t \lf( i u, \partial_s u \ri) \geq - 2 \int_{\annbb} \diff s \diff t \: \lf| \Fk (t) \ri| \lf|  \partial_t u \ri| \lf| \partial_s u \ri|	\\
			 \geq  \int_{\annbb} \diff s \diff t \: (1 - \eps k t) \Fk (t) \lf[ \lf| \partial_t u \ri|^2 + \tx\frac{1}{(1 - 
			 \eps k t)^2} \lf| \partial_s u \ri|^2 \ri],
		}
		where we have used the inequality $ab\leq \frac{1}{2} (\delta a^2 + \delta ^{-1} b^2)$ and the fact that $ \Fk (t) $ is negative for any $ (s,t) \in \ann $ (see Lemma \ref{lem:F prop}). 
%
%
		Combining the above lower bound with \eqref{l b step 1} and dropping the part of the kinetic energy located in $ \ann\setminus\annbb $, we get
		\begin{multline}
 			\label{l b step 2}
			\E_k[u] \geq  \int_{\annbb} \diff s \diff t \: \lf(1 - \eps k t \ri) \Kk(t) \bigg[ \lf| \partial_t u \ri|^2 + \tx\frac{1}{(1- \eps k t)^2} \lf| \partial_s u \ri|^2 \bigg] 	\\
			+ \disp\int_{\ann\setminus\annbb } \diff s \diff t \:  \Fk (t) \partial_t \lf( i u, \partial_s u \ri) + \de \int_{\ann} \diff s \diff t \: \lf(1 - \eps k t \ri) \fk ^2 \lf[ \lf| \partial_t u \ri|^2 + \tx\frac{1}{(1- \eps k t)^2} \lf| \partial_s u \ri|^2 \ri]	\\
			 +  \disp\frac{1}{2 \hex} \int_{\ann} \diff s \diff t \: \lf(1 - \eps k t \ri) \fk ^4 \lf( 1 - |u|^2 \ri)^2,
		\end{multline}
		where $ \Kk (t) $ is the cost function defined in \eqref{K}, for some given $\de$, satisfying \eqref{eq:de}.

		By Proposition \ref{pro:K positive}, $ \Kk $ is positive inside $ \annbb $ and we can thus drop the first term from the lower bound. The main point is now to control the second term. For this purpose we act as in \eqref{l b step 0} to write
		\beq
			\int_{\ann\setminus\annbb} \diff s \diff t \:  \Fk (t) \partial_t \lf( i u, \partial_s u \ri) \geq 2 \int_{\ann\setminus\annbb} \diff s \diff t \: \Fk (t)  \lf|  \partial_t u \ri| \lf| \partial_s u \ri|.
		\eeq
		Then we notice that $ \fk $ is decreasing in $ \ann\setminus\annbb $, so that we can estimate
		\beq
			|\Fk(t)| = - \Fk(t)  = \int_{\teps} ^t d\eta \: \frac{\eta + \alk - \half \eps k \eta ^2}{1-\eps k\eta} \fk ^2 (\eta) \leq C |\log\eps|^2 \fk^2(t),
		\eeq
		because $\Fk$ vanishes at the boundaries of $\ieps$ and $\ieps \setminus \annbk$ has a measure $\OO(|\log \eps|)$ by \eqref{bxi}. Moreover, thanks to the bound \eqref{fal derivative}, we also have 
		\beq
			\lf|\partial_s u \ri| \leq \fk^{-1}(t) \lf| \partial_s \psi(s,t) \ri|,
		\eeq
		\beq
			\lf| \partial_t u \ri| \leq \fk^{-2}(t) \lf| \fk^{\prime}(t) \ri| |\psi(s,t)| + \fk^{-1}(t) \lf| \partial_t \psi(s,t) \ri| \leq  C \fk^{-1}(t) \lf[ |\log\eps|^3 |\psi(s,t)| + \lf| \partial_t \psi(s,t) \ri| \ri],
		\eeq
		and therefore
		\beq
 			\label{l b step 4}
			\int_{\ann\setminus\annbb} \diff s \diff t \:  \Fk(t) \partial_t \lf( i u, \partial_s u \ri) \geq - C |\log\eps|^2 \int_{\ann\setminus\annbb} \diff s \diff t \: \lf[ |\log\eps|^3 |\psi| + |\nabla \psi| \ri] \lf| \nabla \psi \ri|.
		\eeq
		Now, since $\psi(s,t) = \glm(\rv(\eps s, \eps t)) $, $ \rv(s,t) $ standing for the diffeomorphism given by the change into rescaled boundary coordinates, we may use the Agmon estimate \eqref{est grad psi}. It implies, for some finite constant $A>0$,
		\bml{
			\label{l b step 5}
			\int_{\ann\setminus\annbb} \diff s \diff t \: \lf| \nabla \psi \ri|^2 = \int_{\ann\setminus\annbb} \diff s \diff t \: \lf| \lf(\nabla_{s,t} - i \nabla_{s,t} \phi_{\eps} \ri) \glm(\rv(\eps s, \eps t)) \ri|^2 	\\
			\leq 2 \int_{\annt} \diff \rv \:  \lf| \lf( \nabla + i \tx\frac{\aavm}{\eps^2} \ri) \glm \ri|^2 + \frac{C}{\eps^2} |\log\eps|^2 \int_{\annt} \diff \rv \: \lf| \glm \ri|^2 	\\
			\leq C \frac{|\log\eps|^2}{\eps^2} e^{-A \btik} \int_{\annt} \diff \rv \: e^{ \frac{A(1 - r)}{\eps} } \lf\{ \lf| \glm \ri|^2 + \eps^2 \lf| \lf( \nabla + i \tx\frac{\aavm}{\eps^2} \ri) \glm \ri|^2 \ri\}	\\
		= \OO(\eps^{-1 + c_0 A} |\log\eps|^{2+\beta})
		}
		for some finite $ \beta > 0 $, since $ t_{\eps,k} = c_0 |\log\eps| - C\log|\log\eps| $ by \eqref{bxi} and therefore
		\bdm
			e^{-A \bti} = \eps^{c_0 A} |\log\eps|^\beta.
		\edm
		Note that we used the original form of the Agmon estimate \eqref{est grad psi} and the estimate
		\bdm
			\lf| \nabla_{s,t} \phi_{\eps} + \tx\frac{\aavm(\rv(\eps s, \eps t))}{\eps} \ri| \leq  t + o(1) = \OO(|\log\eps|),
		\edm
		following from \eqref{eq:def aeps}, \eqref{eq:diff magn fields} and \eqref{eq:est grad psi rescaled}.
		In a similar way we can also bound
		\beq
			\label{l b step 6}
			\int_{\ann\setminus\annbb} \diff s \diff t \: \lf| \psi \ri|^2 \leq \eps^{-2} e^{-A \bti} \int_{\annbb} \diff \rv \: \exp \lf\{ \tx\frac{A(1 - r)}{\eps} \ri\} \lf| \glm \ri|^2 = \OO(\eps^{-1 + c_0 A} |\log\eps|^\beta).
		\eeq
		Hence putting together \eqref{l b step 4}, \eqref{l b step 5} and \eqref{l b step 6} and using the Cauchy-Schwarz inequality, we have proved 
		\[
		 \int_{\ann\setminus\annbb} \diff s \diff t \:  \Fk (t) \partial_t \lf( i u, \partial_s u \ri) \geq - C\eps ^{-1 + c_0 A} |\log \eps| ^{4+\beta}.
		\]
		Thus, going back to \eqref{l b step 2} and dropping the positive terms 
		\[
		 \E_k [u] \geq - C\eps ^{-1 + c_0 A} |\log \eps| ^{4+\beta}.
		\]
		Taking $c_0$ large enough, this can be made smaller than any power of $\eps$, so in combination with \eqref{splitting} it yields \eqref{eq:low bound disc}.
\end{proof}

We note that to obtain the energy lower bound we have dropped some positive terms (second line of \eqref{l b step 2}), so the proof above actually provides a control of these terms, which is the content of

\begin{cor}[\textbf{Reduced energy bound}]\label{cor:Pan disc pre}\mbox{}\\
Under the  assumptions of Proposition \ref{pro:low bound disc} and for any $ \de $ such that $ 0<\de \leq C |\log\eps|^{-4} $ as $ \eps \to 0 $,
\begin{multline}\label{eq:reduced ener bounds}
	\de \int_{\ann} \diff s \diff t \: \lf(1 - \eps k t \ri) \fk ^2 \bigg[ \lf| \partial_t u \ri|^2 + \tx\frac{1}{(1- \eps k t)^2} \lf| \partial_s u \ri|^2 \bigg] \\ +  \frac{1}{2 \hex} \int_{\ann} \diff s \diff t \: \lf(1 - \eps k t \ri) \fk ^4 \lf( 1 - |u|^2 \ri)^2 
= \OO(\eps |\log \eps|).  
\end{multline}
\end{cor}

Note that \eqref{eq:main density disc} follows from the above in the same way as \eqref{eq:main density generic} followed from \eqref{eq:density generic pre}. In the next subsection we show how to deduce our other results about the GL order parameter from Corollary~\ref{cor:Pan disc pre}.

\subsection{Uniform estimates on the density $ |\glm|^2 $}\label{sec:disc Pan}

Before proceeding let us define $ \anntd $ as the region $ \annd $ (see \eqref{annd} for its definition) in the scaled boundary variables $ (s,t) $
	\beq
		\label{anntd}
		\anntd := \lf\{ (s,t) \in \ann \: : \: \fk(t) \geq \game \ri\}.
	\eeq

	Corollary \ref{cor:Pan disc pre} clearly says that, roughly speaking, $ |u| $ can not differ too much from $ 1 $ inside $ \anntd $ where the density $\fk ^2$ is large enough. In order to extract from this fact a pointwise estimate of $ |u| $ that will lead to \eqref{eq:Pan plus}, another ingredient is needed: 

	\begin{lem}[\textbf{Gradient bound for $|u|$}]
		\label{grad u est: lem}
		\mbox{}	\\
		Let $\Om$ be a disc of radius $R=k ^{-1} > 0 $. For any $ 1 < \hex < \theo^{-1} $ and $ \eps$ sufficiently small, we have
		\beq
			\label{grad u est}
			\lf| \nabla |u|(s,t) \ri| \leq \frac{C |\log\eps|^3}{\fk(t)},
		\eeq
		for any $ (s,t) \in \ann $.
	\end{lem}
	
	\begin{proof}
		From the definitions of $\psi$ and  $u$ in Proposition \ref{pro:low bound disc} and Lemma \ref{lem:split disc}, we immediately have
		\bml{
			\lf| \nabla |u|(s,t) \ri| \leq \fk^{-2}(t) \lf| \fk^{\prime}(t) \ri| \lf| \psi(s,t) \ri| + \fk^{-1}(t) \lf| \nabla_{s,t} \lf| \psi(s,t) \ri| \ri|	\\
			\leq C \fk^{-1}(t) \lf[ |\log\eps|^3 + \lf| \nabla_{s,t} \lf| \psi(s,t) \ri| \ri| \ri],
		}
		where we have used \eqref{fal derivative}. The result is then a consequence of \cite[Theorem 2.1]{Al} or \cite[Eq. (4.9)]{AH} in combination with the diamagnetic inequality (see, e.g., \cite{LL}), which yield	
		\beq
			\lf| \nabla \lf| \glm \ri| \ri| \leq \lf| \lf( \nabla + i \tx\frac{\aavm}{\eps^2} \ri) \glm \ri| = \OO(\eps^{-1})
			\qquad \Longrightarrow \qquad \lf| \nabla_{s,t} \lf| \psi(s,t) \ri| \ri|  = \OO(1).
		\eeq
	\end{proof}

	\begin{proof}[Proof of Theorem \ref{theo:Pan}]
		The argument has been used many times in the literature since its first appearance in \cite{BBH}. For contradiction, suppose that there is a point $ (s_0,t_0) \in \anntd $ such that
		\beq
			\lf| 1 - |u(s_0,t_0)| \ri| \geq \sigme,
		\eeq
		where
		\beq
			\sigme : = \frac{\eps^{1/4} |\log\eps|^{2}}{\game^{3/2}} \ll 1,
		\eeq
		by the assumptions on $ \game $ \eqref{game}. Then by \eqref{grad u est}
		\beq
			\lf| \nabla |u|(s,t) \ri| \leq C \game^{-1} |\log\eps|^3
		\eeq
		inside $ \anntd $. There would therefore exist a ball $ \ba_{\varre}(s_0,t_0) $ centered at $ (s_0,t_0) $ of radius 
		\beq
			\varre = \OO(\game |\log\eps|^{-3} \sigme),
		\eeq
		such that
		\beq
			\lf| 1 - |u(s,t)| \ri| \geq \tx\frac12 \sigme,
		\eeq
		for any $ (s,t) \in  \ba_{\varre}(s_0,t_0) \cap \anntd$. Hence we could bound from below the potential energy appearing in \eqref{eq:reduced ener bounds} as follows\footnote{Actually  the entire ball $ \ba_{\varre}(s_0,t_0) $ might not be contained in $ \ann $, if the point is close to the boundary, but in this case one can arrange this set in such a way that at least half the ball is included.}
		\bml{
 			\int_{\ann} \diff s \diff t \: (1 - \eps t) \fk^4 \lf( 1 - | u |^2 \ri)^2 \geq \int_{\ba_{\varre}(s_0,t_0)} \diff s \diff t \: (1 - \eps t) \fk^4  \lf( 1 - | u |^2 \ri)^2	\\
			\geq C \game^4 \varre^2 \sigme^2 = C  \eps |\log\eps|^2,
		}
		which would contradict the upper bound \eqref{eq:reduced ener bounds}. Hence for any $ (s,t) \in \anntd $
		\beq
			\label{eq:u bound boundary}
			\lf| 1 - |u(s,t)| \ri| \leq \sigme = \OO(\eps^{1/4} \game^{-3/2} |\log\eps|^{2}).
		\eeq
		It is then straightforward to deduce the final result by simply translating the above bound into one on $\glm$ (recall \eqref{def u}).
	\end{proof} 

\subsection{Circulation estimate}\label{sec:circulation}

The main point for the proof of the circulation estimate is a bound on the circulation of $ u $ at the boundary, which is contained in next
	
	\begin{lem}[{\bf Circulation of $ u $ on $ \partial \Omega $}]
		\label{lem:u circulation}
		\mbox{}	\\
		Let $\Om$ be a disc of radius $R=k ^{-1} > 0 $. For any $ 1 < \hex < \theo^{-1} $ and $ \eps$ sufficiently small, 
		\beq
			\label{eq:u circulation}
			\bigg| \int_0^{\frac{|\partial \Omega|}{\eps}} \diff s \: \lf(i u(s,0), \partial_s u(s,0) \ri) \bigg| = \OO(|\log\eps|^3).
		\eeq
	\end{lem}

	\begin{proof}
		The idea of the proof is the same as \cite[Lemma 3.5]{CPRY3}. We use a cut-off function to estimate the circulation integral in terms of a volume integral over a thin layer around the boundary. The width of such a layer is to some extent arbitrary but for later purposes it must be $ o(1) $ in the rescaled normal coordinate $ t $. For the sake of clarity we fix from the outset such a width equal to $ |\log\eps|^{-1} $. 

		Let $ \chi(t) $ be a smooth cut-off function satisfying the properties
		\beq
			\supp(\chi) \subset [0,|\log\eps|^{-1}], 	\qquad \chi \leq 1,		\qquad	\chi(0) = 1,	\qquad	\chi(|\log\eps|^{-1}) = 0,
		\eeq
		\beq
			\lf| \partial_t \chi \ri| = \OO(|\log\eps|).
		\eeq
		Then we can write by means of integration by parts
		\beq
			\int_0^{\frac{|\partial \Omega|}{\eps}} \diff s \: \lf(i u(s,0), \partial_s u(s,0) \ri) = \int_0^{\frac{|\partial \Omega|}{\eps}} \diff s \int_0^{|\log\eps|^{-1}} \diff t \lf\{ \partial_t \chi \lf(i u, \partial_s u\ri) + \chi \partial_t \lf(i u, \partial_s u\ri) \ri\}.
		\eeq
		Hence we can bound the modulus of the r.h.s. of the above expression as
		\bml{
 			\lf| \int_0^{\frac{|\partial \Omega|}{\eps}} \diff s \: \lf(i u(s,0), \partial_s u(s,0) \ri) \ri| \leq \int_0^{\frac{|\partial \Omega|}{\eps}} \diff s \int_0^{|\log\eps|^{-1}} \diff t \lf\{ C |\log\eps| |u| |\nabla u| + 2 \lf| \nabla u \ri|^2 \ri\}	\\
			\leq C \lf[ \frac{\delta}{\eps} + \lf(1 + \frac{|\log\eps|}{\delta} \ri) \lf\| \nabla u \ri\|_{L^2(\ann)}^2 \ri],
		}
		where we have used Cauchy-Schwarz inequality and the estimate \eqref{eq:u bound boundary} and introduced a variational parameter $ \delta $ to optimize over. Then it suffices to exploit the pointwise lower bound \eqref{fal point l u b} together with the estimate \eqref{eq:reduced ener bounds} proven in Corollary \ref{cor:Pan disc pre}, which yields
		\beq
			\lf\| \nabla u \ri\|_{L^2(\ann)}^2  = \OO(\eps |\log\eps|^5),
		\eeq
		where we have chosen $ \de = |\log\eps|^4 $ to meet the condition \eqref{eq:de}, and pick $ \delta = \eps |\log\eps|^3 $ to get the final result.
	\end{proof}

	We can now address the proof of the winding number estimate:
	
	\begin{proof}[Proof of Theorem \ref{teo:circulation}]
		We first note that thanks to the definition of $ \psi $ we have 
		$$ \psi(s,0) = \glm(R,\eps s) e^{-i\phi_{\eps}(s,0)}, $$
		where we have used polar coordinates $ \rv = (r, \vartheta) $ for $ \glm $. Hence the decomposition \eqref{def u} together with the strict positivity of $ \fk(0) $ by \eqref{fal point l u b} and the estimate of $ |u| $  \eqref{eq:u bound boundary} imply that $ \glm $ does not vanish on $ \partial \Omega $ and its winding number is thus well defined. Note that the positivity of $ \fk(0) $ requires the condition $ 1 < b < \theo^{-1} $.
		Moreover $ \fk(0) $ is separated from $ 0 $ independently of $ \eps $ again by \eqref{fal point l u b}, so that \eqref{eq:u bound boundary} holds true with $ \game = \mathrm{const} $.
		
		Before computing the contribution to the winding number due to the optimal phase $ \alk $, we first isolate the leading term generated by the change of gauge \eqref{phieps}: recalling again that $ \psi(s,0) = \glm(\rv(\eps s,0))  e^{-i\phi_{\eps}(s,0)} $ (see Proposition \ref{replacement: pro}), one has
		\bml{
			\label{gauge phase}
			2 \pi \deg\lf(\glm, \partial \Omega\ri) - 2\pi \deg\lf(\psi, \partial \Omega \ri) =  \int_{0}^{\frac{2\pi}{\eps}} \diff s \: \partial_{s} \phi_{\eps}(s,0) = \frac{1}{\eps^2} \int_{\partial  \Omega} \diff \sigma \: {\bf e}_{\vartheta} \cdot \aavm	 - \deps \\
			=  \frac{1}{\eps^2} \int_{\Omega} \diff \rv \: \curl \aavm - \deps.
			}
			Now by the elliptic estimate \eqref{eq: elliptic est}
			\beq
				 \lf\| \curl \aavm - 1 \ri\|_{L^{1}(\ann)} = \OO(\eps^2 |\log\eps|),
			\eeq
			while the Agmon estimate for $ \aavm $ \cite[Eq. (12.10)]{FH1} implies 
			\bdm
				\lf\| \nabla (\curl \aavm - 1) \ri\|_{L^1(\Omega\setminus \ann)} = \OO(\eps^{\infty}),
			\edm
			implies
			\bml{
				 \lf\| \curl \aavm - 1 \ri\|_{L^{1}(\Omega\setminus \ann)} \leq C  \lf\| \curl \aavm - 1 \ri\|_{L^{2}(\Omega\setminus \ann)} 	\\
				 \leq C \lf\| \nabla (\curl \aavm - 1) \ri\|_{L^1(\Omega\setminus \ann)} = \OO(\eps^{\infty}),
			}
			via Sobolev inequality. Altogether we can thus replace $ \curl \aavm $ with $ 1 $ in \eqref{gauge phase}, so obtaining
			\beq
				2 \pi \deg\lf(\glm, \partial \Omega\ri) - 2\pi \deg\lf(\psi, \partial \Omega \ri) =  \frac{\pi R^2}{\eps^2} + \OO(|\log\eps|).
			\eeq
			
		Now it remains to estimate the contribution to the winding number of $ \psi $:
		\bml{
 			\label{eq:circulation 1}
 			 2\pi \deg\lf(\psi, \partial \Omega \ri) 
			 = -i \int_0^{\frac{2\pi}{\eps}} \diff s \: \frac{|u(s,0)| e^{i(\alk + \eps\deps) s}}{u(s,0)} \partial_{s} \lf( \frac{u(s,0) e^{-i(\alk + \eps\deps) s}}{|u(s,0)|} \ri)	\\
			= - \frac{2 \pi\lf(  \alk + \eps \deps \ri)}{\eps}  \lf(1 + \OO(\eps^{1/4}|\log\eps|^2) \ri) - i  \int_0^{\frac{2\pi}{\eps}} \diff s \: \frac{|u(s,0)|}{u(s,0)} \partial_{s} \lf( \frac{u(s,0)}{|u(s,0)|} \ri),
		}
		where we have made use of \eqref{eq:u bound boundary}, i.e., 
		\beq
			\label{eq:mod u boundary}
			|u(s,0)| = 1 + \OO(\eps^{1/4} |\log\eps|^2).
		\eeq 
		Now to complete the proof it remains to bound the second term on the r.h.s. of \eqref{eq:circulation 1}, but again by \eqref{eq:mod u boundary}, we get
		\bdm
			\bigg| \int_0^{\frac{2\pi}{\eps}} \diff s \: \frac{|u(s,0)|}{u(s,0)} \partial_{s} \lf( \frac{u(s,0)}{|u(s,0)|} \ri) \bigg| = \lf(1 + \OO(\eps^{1/4}|\log\eps|^2) \ri) \bigg| \int_0^{\frac{2\pi}{\eps}} \diff s \: \lf(i u(s,0), \partial_s u(s,0) \ri) \bigg|,
		\edm	
		and therefore Lemma \ref{lem:u circulation} provides the result.
	\end{proof}

%
%

\appendix

\section{Appendix}

\subsection{The one-dimensional Schr\"{o}dinger operator $ \hkal $}\label{1d So: sec}

We start  by stating standard results about the ground state of the operator $ \hkal $ defined in \eqref{eq:Hkal}:

	\begin{pro}[\textbf{Ground state of $ \hkal $}]
		\label{Ha gs: pro}
		\mbox{}	\\
		For any given $ \alpha \in \R $, $ k \geq 0 $ and $ \eps $ small enough, the operator $\hkal $ is positive and self-adjoint in  
		$$ \mathscr{H} : =  L^2(I_{\eps}, (1 - \eps k  t) \diff  t). $$
		Its normalized ground state $ \phikal $, i.e., the lowest eigenstate solving
		\beq
			\hkal \phikal = \mue \phikal,
		\eeq
		is unique up to the multiplication by a constant phase factor, which can be chosen in such a way that it is real and strictly positive. Moreover $ \phikal \in C^{\infty}(\ie) $ and it satisfies the following bounds
		\beq
			\label{phi prime est}
			\lf\| \phikal^{\prime} \ri\|_{L^{\infty}(I_{\eps})} = \OO(|\log\eps|^{5/2}),
		\eeq
		\beq
			\label{phi decay est}
			\phikal( t) \leq C \exp \lf\{ - \half  t^2 \ri\},
		\eeq
		for some finite constant $ C $.
	\end{pro}

	\begin{proof}
		The first part of the Proposition can be proven by exploiting standard techniques of operator theory. The bound on the derivative of $ \phikal $ can be obtained by simply integrating the eigenvalue equation:
		\bdm
			\int_{0}^{ t} \diff  \eta \: \partial_\eta \lf[ \lf(1 - \eps k \eta \ri) \phikal^{\prime}( \eta) \ri] = \int_{0}^{ t} \diff  \eta \: \lf(1 - \eps k \eta \ri) \lf( \potkal ( \eta) - \mue \ri) \phikal( \eta),
		\edm
		which yields by a trivial application of Cauchy-Schwarz inequality and normalization of $ \phikal $
		\beq
			\lf|(1-\eps k t )\phikal^{\prime}( t)\ri| \leq \bigg[ \int_0^{ t} \diff  \eta \: \potkal^2( \eta) \bigg]^{1/2} + \mue \sqrt{ t},
		\eeq
		and thus \eqref{phi prime est}, via the trivial bound $ \mue = \OO(|\log\eps|^2) $, which can be easily obtained by evaluating the energy of a normalized constant function.

		The decay estimate \eqref{phi decay est} is proven by exploiting the resolvent of the harmonic oscillator $ \hosc : = - \partial_ t^2 +  t^2 $ on the real line, i.e., for any $ \lambda > -1 $,
		\beq
			\lf( \hosc + \lambda \ri)^{-1}( t ;  t^{\prime}) = \frac{1}{2\sqrt{\pi}} \int_0^{1} \diff \nu \: \frac{\nu^{2\lambda - 1/2}}{\sqrt{1 - \nu^2}} \exp \lf\{ - \tx\frac{1-\nu^2}{2(1+\nu^2)} \big(  t^2 + { t^{\prime}}^2 \big) + \frac{2\nu}{1-\nu^2}  t  t^{\prime} \ri\}.
		\eeq	
		To this purpose we first regularize $ \phikal $ in order to associate it with a function in the domain of $ \hosc $. A simple way to do that is multiply $ \phikal $ by a smooth function $ \chi_{\eps} \leq 1 $ with the following properties
		\beq
			\label{chieps}
			\chi_{\eps}( t) : = 
			\begin{cases}
				0,	&	\mbox{if }  t = 0, c_0|\log\eps|,	\\
				1,	&	\mbox{for } \eps^{\gamma} \leq  t \leq c_0|\log\eps| - \eps^{\gamma},
			\end{cases}
		\eeq
		for some $ \gamma \geq 1 $. As a straightforward consequence $ \tilde\phi : = \chi_{\eps} \phikal $ is approximately normalized, i.e., $ \| \tilde\phi \lf. \ri\|_2 = 1 + o(1) $, and satisfies the following differential equation in $ I_{\eps} $
		\beq
			- \tilde\phi^{\prime\prime} + ( t+\alpha)^2 \tilde\phi = \mue \tilde\phi + \lf[ ( t + \alpha)^2 - \potkal( t) \ri] \tilde\phi - \tx\frac{\eps k}{1 - \eps t} \chi_{\eps} \phikal - 2 \chi_{\eps}^{\prime} \phikal^{\prime} - \chi_{\eps}^{\prime\prime} \phikal.
		\eeq 
		Calling $ \Phi( t) $ the r.h.s. of the above expression, we can apply the resolvent of $ \hosc $ to get
		\beq
			\tilde\phi( t) = \int_{\R} \diff  t^{\prime} \: \lf( \hosc \ri)^{-1}( t, t^{\prime} + \alpha) \Phi( t^{\prime}).
		\eeq
		Moreover exploiting the bound on $ \phi^{\prime}_{\alpha} $ and choosing $ \chi_{\eps} $ smooth enough, e.g., so that $ \lf\|\chi_{\eps}^{\prime}\ri\|_{\infty} \leq \OO(\eps^{-\gamma}) $ etc., we can apply Cauchy-Schwarz inequality to extract an upper bound to $ \tilde\phi $ of the form \eqref{phi decay est}, which immediately translates into an estimate for $ \phikal $ via \eqref{chieps}.		
	\end{proof}

	Most of the properties of the ground state energy and wave function of $ \hkal $ are obtained by a comparison with the shifted harmonic oscillator on the half-line with Neumann conditions at the boundary $  t = 0 $:
	\beq
		\hosc_{\alpha} : = - \partial_{ t}^2 + ( t + \alpha)^2,
	\eeq
	acting on $ L^2(\R^+,\diff  t) $ (notice the different measure). We denote by $ \muosc $ its ground state energy and refer to \cite[Chapter 3.2]{FH} for a detailed analysis of its properties, that we briefly recall here:
	\beq
		\mu^{\mathrm{osc}}(-\infty) = \mu^{\mathrm{osc}}(0) = 1,	\qquad \muosc < 1,	\quad \mbox{for } \alpha <0,
	\eeq
	\beq
		\mu^{\mathrm{osc}}(+\infty) = +\infty,	\qquad 	 \muosc > 1,	\quad \mbox{for } \alpha >0.
	\eeq
	Moreover $ \muosc $ is monotonically increasing for $ \alpha \geq 0 $, whereas it admits a unique minimum at $ - \sqrt{\theo} $ given by 
	\beq
		\label{eq:min muosc}
		\min_{\alpha \in \R} \muosc = \mu^{\mathrm{osc}}\lf(- \sqrt{\theo}\ri) = \theo < 1.
	\eeq

	\begin{pro}[\textbf{Ground state energy $ \mue $}]
		\label{mue: pro}
		\mbox{}	\\
		For any given $ \alpha \in \R $, $ k > 0 $ and $ \eps $ small enough, $ \mue $ is smooth function of $ \alpha $ and
		\beq
			\label{mue diff}
			\mue - \muosc = \OO(\eps |\log\eps|^3). 
		\eeq
	\end{pro}
	
	\begin{proof}
		It is easy to verify that $ \hkal $ is an analytic family of operators of type (A) \cite[Definition at p. 16]{RS4}: the difference $ \hkal - \tilde{H}_{\alpha} $, with
		\beq
			\tilde{H}_{\alpha} : = - \partial_{ t}^2 + \tx\frac{\eps k}{1-\eps k t} \partial_ t + \lf(  t + \alpha \ri)^2,
		\eeq
		is indeed a bounded operator with norm (see below) 
		\bdm
			\|  \hkal - \tilde{H}_{\alpha} \| = \OO(\eps|\log\eps|^3)
		\edm
		and therefore $ \hkal - \tilde{H}_{\alpha} $ is $ \tilde{H}_{\alpha} $-bounded. On the other hand $ \tilde{H}_{\alpha} $ is an analytic family of operators in the sense of Kato since it is precisely the shifted harmonic operator on $ L^2(I_{\eps},(1-\eps k t) \diff  t) $. In conclusion $ \mue $ is an isolated non-degenerate eigenvalue of $ \hkal $ and thus it must be a smooth (in fact analytic) function of $ \alpha $.


		The estimate of the difference $ \mue - \muosc $ is a consequence of the estimate \eqref{potkal est},
		which implies the operator inequality 
		\bdm
			\|  \hkal - \tilde{H}_{\alpha} \| \leq C\eps|\log\eps|^3 
		\edm
		and therefore 
		\bdm
			\mue - \mu_{0}(\alpha) = \OO(\eps|\log\eps|^3),
		\edm
		 where  $ \mu_{0}(\alpha) $ stands for the ground state energy of $ \tilde{H}_{\alpha} $. However since $ \tilde{H}_{\alpha} $ coincides with the shifted harmonic oscillator on $ L^2(I_{\eps},(1-\eps k t) \diff  t) $, some simple upper and lower bounds allow to conclude that 
		\bdm
			\mu_0(\alpha) - \muosc = \OO(\eps|\log\eps|),
		\edm 	
		which in turn yields the result.

	\end{proof}

	\begin{cor}[\textbf{Ground state energy $ \mue $, continued}]
		\label{mue min: cor}
		\mbox{}	\\
		For any $ k > 0 $ and $ \eps $ small enough
		\beq
			\label{infa}
			\inf_{\alpha \in \R} \mue = \theo + \OO(\eps|\log\eps|^3),
		\eeq
		and for any $ \hex $ satisfying $ 1 < \hex < \theo^{-1} $, there exist 
		\beq
			\label{al up down}
			- \infty < \aldown_{1}(k,\hex) \leq \alup_{1}(k,\hex) <  \aldown_{2}(k,\hex) \leq \alup_{2}(k,\hex) < 0,
		\eeq
		with $ \alup_{i} - \aldown_{i} = o(1) $, $ i = 1,2 $, and $ \aldown_2 - \alup_1 > C(k,b) > 0 $, so that
		\beqn
				\hex^{-1} > \mue,		 &\quad	& \mbox{for any }  \alpha \in (\alup_1, \aldown_2), \nonumber	\\
				\label{cond hex}
				\hex^{-1}\leq \mue, 	 & \quad & 		\mbox{for any } \alpha \in (-\infty,  \aldown_{1}] \mbox{ or }  \alpha \in [\alup_2, \infty).
		\eeqn	
	\end{cor}

	\begin{proof}
		All the results are simple consequence of \eqref{mue diff} combined with the properties of $ \muosc $ proven in \cite[Chapter 3.2]{FH}. Note that a priori the equation $ \hex^{-1} = \mue $ might have more that two solutions, unlike the equation $ \hex^{-1} = \muosc $, which is the reason why in general $ \aldown_i \neq \alup_i $. The existence of the interval $(\alup_1, \aldown_2) $ of size $ \OO(1) $ where  $ \hex^{-1} > \mue $ is however implied by the properties of $ \muosc $.
	\end{proof}

\subsection{Technical estimates of the one-dimensional density profiles}\label{sec:app est}

We prove in this section the pointwise upper and lower bounds to the density $ \fal $ minimizing the 1D functional $ \fonekal $. 

\begin{proof}[Proof of Proposition \ref{pro:point est fal}]
	Although similar arguments are used very often in literature (see, e.g., \cite{CPRY3}), we discuss them in details for the sake of completeness. Throughout the proof we use the conditions $ \hex \in (1,\theo^{-1}) $ and $ \al \in (\alup_2,\aldown_1) $, which guarantee the positivity of $ \fal $ via Corollary \ref{nontrivial: cor}.

	Both bounds are essentially consequences of the maximum principle but, in order to make any comparison possible, we first have to extend the density $ \fal $ to a smooth function $ \tilde f $ on the whole semi-axis $ \R^+ $, in such a way that a differential inequality is still satisfied, namely 
\bdm
	 - \tilde f^{\prime\prime} + \tx\frac{\eps k}{1 - \eps k t} f^{\prime} + \pot f \geq \frac{1}{\hex}(1 - \tilde f^2) \tilde f.
\edm
In other words $ \tilde f  $ must be a weak supersolution to the variational equation \eqref{var eq fal}. This is of course doable in several ways but for the sake of clarity we propose here an explicit $ C^2$-continuation of $ \fal $: it suffices to set for some $ a> 0 $ and any $  t \geq  \teps $
	\beq
		\tilde f( t) : = \lf( f( \teps) + a( t -  \teps)^2 \ri) \exp \lf\{ - \tx\frac{1}{2} \lf(  t -  \teps \ri)^2 \ri\},
	\eeq
	and notice that
	\bdm
		\tilde f ( \teps) = \fal( \teps),	\qquad	\tilde f^{\prime}( \teps) = 0 = \fal^{\prime}( \teps),
	\edm
	thanks to Neumann boundary conditions. Moreover it is easy to verify that by picking
	\beq
		\label{constant a}
		a = \tx\frac{1}{2} \lf[\potkal( \teps) - \frac{1}{\hex} \lf(1 - \fal^2( \teps) \ri) + 1 \ri] \fal( \teps),
	\eeq
	one also has
	\bdm
		\tilde f ^{\prime\prime}( \teps) = \fal^{\prime\prime}( \teps),
	\edm
	i.e., $ \tilde f \in C^2(\R^+) $.

	Let us now discuss the lower bound. We introduce the function		
	\beq
		w( t) : = \fal(0) \exp\lf\{ - \tx\frac{1}{2}\lf(  t + \sqrt 2 \ri)^2 \ri\},
	\eeq
	and note that 
	\beq
		w(0) < \fal(0), 	\qquad \lim_{ t \to \infty} w( t) = \lim_{ t \to \infty} \tilde f( t) = 0.
	\eeq
	Moreover setting
	\beq
		u( t) : = w( t) - \tilde f ( t),
	\eeq
	we have
	\bml{ 
 		\label{ineq a 0}
		- u^{\prime\prime} + \Big[ \big(  t + \sqrt 2 \big)^2 -  1 \Big] u 	\\	
		= 
		\begin{cases}
			 \frac{\eps k}{1 - \eps k t}  \fal^{\prime}+ \lf[ \potkal(t) - \big(  t + \sqrt 2 \big)^2 + 1 - \frac{1}{\hex} (1 - \fal^2) \ri] \fal,	&	\mbox{if }  t \in [0, \teps],	\\
			\lf\{ - \big(  t + \sqrt 2 \big)^2 + \lf(  t - \teps\ri)^2 + \frac{2  a [1 - 2( t -  \teps)]}{\fal( \teps) + a \lf(  t - \teps\ri)^2} \ri\} \tilde f, 	&	\mbox{if }  t \geq  \teps.
		\end{cases}
	}
	Now for $  t \leq \teps$ we exploit next Lemma \ref{fal derivative: lem}, which combined with the fact that $ \fal $ is decreasing for $  t \geq \max[0, - \alpha + \frac{1}{\sqrt{b}} + \OO(\eps)] $ (Proposition \ref{min fone: pro}) yields that $ \fal^{\prime} \leq C \fal $ for some finite $ C $. Indeed for finite $  t $, \eqref{fal derivative} should be used in combination with \eqref{fal point l u b}, implying $ | \fal^{\prime} | \leq C' \leq C \fal $. For $  t \geq \max[0, - \alpha + \frac{1}{\sqrt{b}} + \OO(\eps)] $, $ \fal^{\prime} $ is negative and the inequality holds true with any positive constant. Hence we get for $  t \in [0, \teps] $ (recall that $ \alpha \leq 0 $)
	\bml{
 		\label{ineq a 1}
		 \tx\frac{\eps k}{1 - \eps k t}  \fal^{\prime}+ \lf[ \potkal(t) - \big(  t + \sqrt 2 \big)^2 + 1 - \frac{1}{\hex} (1 - \fal^2) \ri] \fal	\\
		 \leq \lf[ -2 \big(\sqrt 2 - \alpha \big)  t - 2 + \alpha^2 + 1 + \OO(\eps) \ri] \fal \leq 0,
	}	
	while for $  t \geq \teps$
	\bml{
 		\label{ineq a 2}
 		 - \big(  t + \sqrt 2 \big)^2 + \lf(  t - \teps\ri)^2 + \disp\frac{2 a [1 - 2( t -  \teps)]}{\fal( \teps) + a \lf(  t - \teps\ri)^2} \leq  -2 \big( \sqrt 2 - \teps\big)  t - 1 +  \teps^2 + \half  \potkal( \teps) 	\\
		 \leq - \half  \teps^2 \lf(1 + \OO(|\log\eps|^{-1}) \ri) \leq 0,
	}
	where we have used that $ 0 \leq a \leq \frac{1}{2} \lf[\potkal( \teps) + 1 \ri] \fal( \teps) $  by \eqref{constant a}. 
	
	Putting together \eqref{ineq a 1} and \eqref{ineq a 2} with \eqref{ineq a 0}, we obtain
	\beq
		- u^{\prime\prime} \leq - \Big[ \big(  t + \sqrt 2 \big)^2 -  1 \Big] u \leq - u.
	\eeq
	Hence $ u $ is subharmonic where $ u > 0 $, but then in that region $ u $ should reach its maximum value at the boundary, where $ u = 0 $ (recall that $ u \leq 0 $ at $  t = 0, \infty $). {Therefore} the region must be empty and $ u \leq 0 $ everywhere. To conclude the proof is then sufficient to use the strict positivity of $ \fal $ at the origin, since it is the ground state of a 1D Schr\"{o}dinger operator\footnote{One should actually show that there exists a constant $ C > 0 $ independent of $ \eps $ such that $ \fal(0) \geq C $. This is not a straightforward consequence of $ \fal $ being a ground state, since its value at the origin might depend on $ \eps $. However a pointwise estimate of the difference between $ \fal $ and the minimizer of $ \fone_{0,\al} $ with $ \eps = 0 $ can be easily derived, showing that the two functions differ by $ o(1) $. This in turn yields the desired estimate. We omit further details for the sake of brevity.}.

	For the upper bound proof we proceed in the same way by showing that 
	\bdm
		W( t) : = C \exp \lf\{ -\frac{1}{2} ( t + \al)^2 \ri\} 
	\edm
	 is a supersolution for a large enough constant $ C ${. Denoting} again by $ U( t) $ the difference $ W( t) - \tilde f( t) $, one has (making use of \eqref{fal derivative})
	\beq
		- U^{\prime\prime} \geq - \lf[ ( t + \al)^2  - 1 \ri] U \geq - U,
	\eeq
	if $  t \geq |\alpha| + \sqrt{2} $. Now if the constant $ C $ is taken so that $ u(|\alpha| + \sqrt{2}) \geq 0 $, which can always be done since $ \fal \leq 1 $, then by a superharmonicity argument one can conclude that $ U \geq 0 $ in the whole region $  t \geq |\alpha| + \sqrt{2} $. The estimate then  easily extends to the region $  t \leq |\alpha| + \sqrt{2} $, where $ W \geq 1 $.
\end{proof}

	We conclude the section with a technical estimate on the derivative of $ \fal $:
	
	\begin{lem}[\textbf{Gradient bounds for $\fkal$}]
		\label{fal derivative: lem}
		\mbox{}	\\
		For any $ \hex \in (1,\theo^{-1}) $, $ \al \in (\alup_2,\aldown_1) $, $ k > 0 $ and $ \eps $ sufficiently small, there exists a finite constant $ C $ such that
		\beq
			\label{fal derivative}
			\lf| \fal^{\prime}( t) \ri| \leq C
			\begin{cases}
				1,	&	\mbox{for }  t \in \lf[ 0, |\alpha| + \tx\frac{2}{\sqrt{\hex}} \ri],	\\
				|\log\eps|^3 \fal( t),		&	\mbox{for }  t \in \lf[ |\alpha| + \tx\frac{2}{\sqrt{\hex}}, c_0|\log\eps| \ri].
			\end{cases}
		\eeq
		In the case $ k = 0, \eps = 0 $ the above estimates become
		\beq
			\label{eq:fOal derivative}
			\lf| \fOal^{\prime}( t) \ri| \leq C
			\begin{cases}
				1,	&	\mbox{for }  t \in \lf[ 0, |\alpha| + \tx\frac{2}{\sqrt{\hex}} \ri],	\\
				t (t + \al)^2 \fOal( t),		&	\mbox{for }  t \geq |\alpha| + \tx\frac{2}{\sqrt{\hex}}.
			\end{cases}
		\eeq
	\end{lem}

	\begin{proof}
		The result can be easily obtained by integrating the variational equation \eqref{var eq fal}:
		\bdm
			- \tx\frac{1}{1 - \eps k t} \partial_t \lf[ (1 - \eps k t) \fal^{\prime} \ri] + \potkal \fal = \tx\frac{1}{\hex} (1 - \fal^2) \fal,
		\edm
		which alternatively yields (thanks to Neumann boundary conditions)
		\beq
			- (1 - \eps k t) \fal^{\prime}( t) = \int_0^{ t} \diff \eta \: (1 - \eps k \eta) \lf[  \tx\frac{1}{\hex} (1 - \fal^2) - \potkal(\eta) \ri] \fal,
		\eeq
		\beq
			(1 - \eps k t) \fal^{\prime}( t) = \int_{ t}^{\teps} \diff \eta \: (1 - \eps k \eta) \lf[  \tx\frac{1}{\hex} (1 - \fal^2) - \potkal(\eta) \ri] \fal.
		\eeq
		By taking the absolute value of the first identity for any finite $  t $ and exploiting that $ |\fal| \leq 1 $, we get the first inequality. On the other hand for $  t \geq |\alpha| + \tx\frac{2}{\sqrt{\hex}} $, $ \potkal \geq \frac{1}{\hex} $ so that the r.h.s. of the second identity is negative, which in particular implies that $ \fal $ is decreasing there, i.e., $ \fal(\eta) \leq \fal( t) $ for any $ \eta \geq  t $. Taking again the absolute value of both sides and estimating roughly $ \eta $ in the integral with its maximum value, i.e., $ \teps $, we then obtain the second estimate.
			
		The result in the case $ k = 0 $ and $ \eps = 0 $ can be obtained exactly in the same way.
	\end{proof}	


\end{document}